\documentclass[a4,11pt]{article}
\usepackage[utf8]{inputenc}
\usepackage{mathtools}
\usepackage{amsmath}
\usepackage{amsthm}
\usepackage{amssymb} 
\usepackage{graphicx} 
\usepackage{subcaption}%for subfigure
\usepackage{mwe} %for minipage
\usepackage[letterpaper,margin=1in]{geometry}
\usepackage{xcolor}
\usepackage[hidelinks]{hyperref}
\usepackage[ruled,vlined,linesnumbered]{algorithm2e}
\usepackage{tikz}

\usepackage{mdframed}
\makeatletter
\patchcmd{\@algocf@start}% <cmd>
  {-1.5em}% <search>
  {0pt}% <replace>
{}{}% <success><failure>
\makeatother

\usepackage[displaymath]{lineno}
\newcommand*\patchAmsMathEnvironmentForLineno[1]{%
  \expandafter\let\csname old#1\expandafter\endcsname\csname #1\endcsname
  \expandafter\let\csname oldend#1\expandafter\endcsname\csname end#1\endcsname
  \renewenvironment{#1}%
     {\linenomath\csname old#1\endcsname}%
     {\csname oldend#1\endcsname\endlinenomath}}%
\newcommand*\patchBothAmsMathEnvironmentsForLineno[1]{%
  \patchAmsMathEnvironmentForLineno{#1}%
  \patchAmsMathEnvironmentForLineno{#1*}}%
\AtBeginDocument{%
  \patchBothAmsMathEnvironmentsForLineno{equation}%
  \patchBothAmsMathEnvironmentsForLineno{align}%
  \patchBothAmsMathEnvironmentsForLineno{flalign}%
  \patchBothAmsMathEnvironmentsForLineno{alignat}%
  \patchBothAmsMathEnvironmentsForLineno{gather}%
  \patchBothAmsMathEnvironmentsForLineno{multline}}

\newtheorem{theorem}{Theorem}
\newtheorem{lemma}[theorem]{Lemma}
\newtheorem{corollary}[theorem]{Corollary}
\newtheorem{definition}[theorem]{Definition}
\numberwithin{theorem}{section}

\let\P\relax
\DeclareMathOperator{\P}{\mathbb{P}}
\DeclareMathOperator{\E}{\mathbb{E}}
\DeclareMathOperator{\N}{\mathbb{N}}
\DeclareMathOperator{\R}{\mathbb{R}}
\DeclareMathOperator{\poly}{poly}
\DeclareMathOperator{\argmin}{argmin}

\title{Faster Cut Sparsification of Weighted Graphs\footnote{This project has received funding from the European Research Council (ERC) under the European Union's Horizon 2020 research and innovation programme (grant agreement No 947702) and is supported by the Austrian Science Fund (FWF): P 32863-N.}}

\author{Sebastian Forster, Tijn de Vos}
\date{Department of Computer Science\\
University of Salzburg, Austria}

\begin{document}

\maketitle
\abstract{A cut sparsifier is a reweighted subgraph that maintains the weights of the cuts of the original graph up to a multiplicative factor of $(1\pm\epsilon)$. This paper considers computing cut sparsifiers of weighted graphs of size $O(n\log (n)/\epsilon^2)$. Our algorithm computes such a sparsifier in time $O(m\cdot\min(\alpha(n)\log(m/n),\log (n)))$, both for graphs with polynomially bounded and unbounded integer weights, where $\alpha(\cdot)$ is the functional inverse of Ackermann's function. This improves upon the state of the art by Bencz\'ur and Karger (SICOMP 2015), which takes $O(m\log^2 (n))$ time. For unbounded weights, this directly gives the best known result for cut sparsification. Together with preprocessing by an algorithm of Fung et al.\ (SICOMP 2019), this also gives the best known result for polynomially-weighted graphs. Consequently, this implies the fastest approximate min-cut algorithm, both for graphs with polynomial and unbounded weights. In particular, we show that it is possible to adapt the state of the art algorithm of Fung et al.\ for unweighted graphs to weighted graphs, by letting the partial maximum spanning forest (MSF) packing take the place of the Nagamochi-Ibaraki (NI) forest packing. MSF packings have previously been used by Abraham at al.\ (FOCS 2016) in the dynamic setting, and are defined as follows: an $M$-partial MSF packing of $G$ is a set $\mathcal{F}=\{F_1, \dots, F_M\}$, where $F_i$ is a maximum spanning forest in $G\setminus \bigcup_{j=1}^{i-1}F_j$. Our method for computing (a sufficient estimation of) the MSF packing is the bottleneck in the running time of our sparsification algorithm. 
}

\section{Introduction}
In many applications, graphs become increasingly large, hence storing and working with such graphs becomes a challenging problem. One strategy to deal with this issue is graph sparsification, where we model the graph by a sparse set of (reweighted) edges that preserve certain properties. Especially because the aim is to work with large input graphs, this process should be efficient with respect to the graph size. Among the different types of graph sparsifiers, there are spanners (preserving distances, resistance sparsifiers (preserving effective resistances, see e.g.~\cite{DKW15}), see e.g.\ \cite{PS89,AGDJS93,BS07, EN16}), cut sparsifiers (preserving cuts, see e.g.\ \cite{BK96, BK15, FHHP19}), and spectral sparsifiers (preserving Laplacian quadratic forms, see e.g.\ \cite{ST11, SS11,KX16,LS17}). This paper focuses on cut sparsifiers, as first introduced by Benczúr and Karger in \cite{BK96}. We say that a (reweighted) subgraph $H\subseteq G$ is a $(1\pm \epsilon)$-\emph{cut sparsifier} for a weighted graph $G$ if for every cut $C$, the total weight $w_H(C)$ of the edges of the cut in $H$ is within a multiplicative factor of $1\pm\epsilon$ of the total weight $w_G(C)$ of the edges of the cut in $G$.  

The main approach to compute cut sparsifiers uses the process of \emph{edge compression}: each edge $e\in E$ is part of the sparsifier with some probability $p_e$, and if selected obtains weight $w(e)/p(e)$. It is immediate that such a scheme gives a sparsifier in expectation, but it has to be shown that the result is also a sparsifier with high probability. The main line of research has been to select good \emph{connectivity estimators} $\lambda_e$ for each edge such that sampling with $p_e\sim 1/\lambda_e$ yields a good sparsifier. The simplest such result is by Karger \cite{karger99}, where we sample uniformly with each $\lambda_e$ equal to the weight of the min cut. Continuing along these lines are parameters as: edge connectivity \cite{FHHP19}, strong connectivity \cite{BK96,BK15}, electrical conductance \cite{SS11}, and Nagamochi-Ibaraki (NI) indices \cite{NI92computing,NI92linear,FHHP19}. The challenge within the approach of edge compression is to find a connectivity estimator that results in a sparse graph, but can be computed fast. 

For weighted graphs, there are roughly three regimes for sparsification. The first regime consists of cut sparsifiers of size $O(n\log^2(n)/\epsilon^2)$. Fung, Hariharan, Harvey, and Panigrahi \cite{FHHP11,FHHP19} show that sparsifiers of this asymptotic size can be computed in linear time for polynomially-weighted graphs. For this they introduce a general framework of cut sparsification with a connectivity estimator, see Section~\ref{subsc:framework}. For unbounded weights, Hariharan and Panigrahi \cite{HP10} give an algorithm to compute a sparsifier of size $O(n\log^2 (n)/\epsilon^2)$ in time $O(m\log^2(n)/\epsilon^2)$.

The second regime consists of cut sparsifiers of size $O(n\log (n)/\epsilon^2)$. Benczúr and Karger \cite{BK96, BK15} show that these can be computed in time $O(m\log ^2 (n))$ for polynomially-weighted graphs, and in time $O(m\log^3 (n))$ for graphs with unbounded weights. Note that these results can be optimized by preprocessing with the algorithms for the first regime. 

A third regime, consists of sparsifiers of size $O(n/\epsilon^2)$. The known constructions in this regime yield \emph{spectral} sparsifiers, which are more general than cut sparsifiers. Spectral sparsification was first introduced by Spielman and Teng in \cite{ST11}. It considers subgraphs that preserve Laplacian quadratic forms. Lee and Sun \cite{LS17} give an algorithm for finding $(1\pm\epsilon)$-spectral sparsifiers of size $O(n/\epsilon^2)$ in time $O(m\cdot\text{poly}(\log (n),1/\epsilon))$. Analyzing their results, we believe that the poly-logarithmic factor contributes at least a factor of $\log^{10}(n)$. While this is optimal in size, both for spectral sparsifiers \cite{BSS12} and cut sparsifiers \cite{ACK+16}, it is not in time. 

In this paper, we improve on the results in the second regime, both for graphs with polynomially bounded and unbounded weights\footnote{See Section~\ref{sc:model} for our assumptions on the computational model in case of unbounded weights.}. For an overview of the previous best running times and our results, see Figure~\ref{fig:table}.
%In \cite{FHHP19}, it is stated for sparsification algorithm on unweighted graphs that both the bounds on the running time and the bound on the size of the sparsifier hold in expectation. We believe that a tightening of the analysis can show that these bounds hold with high probability. In this paper we show that this is the case for our algorithm for weighted graphs. 
%expect to be useful outside this application.... 
%While maintaining the bound on the size, we improve on the time bound. Moreover, we prove our time and size bounds `with high probability', rather than `in expectation' of Theorem~\ref{thm:FHHP}. 
We present our sparsification algorithm in Section~\ref{sc:ouralg}, with the special treatment of unbounded weights in Section~\ref{sc:arbweights}. Our algorithm improves on the algorithm of Benczúr and Karger \cite{BK96,BK15} for bounded weights, which has been unchallenged for the last 25 years. It also improves on the algorithm of \cite{HP10} for unbounded weights, which has been unchallenged for the last 10 years. We obtain the following theorem, where $\alpha(\cdot)$ refers to the functional inverse of Ackermann's function, for a definition see e.g.\ \cite{tarjan75}. For any realistic value $x$, we have $\alpha(x)\leq 4$. 

\begin{theorem}
\label{thm:resultgeneral}
There exists an algorithm that, given a weighted graph $G$ and a freely chosen parameter $\epsilon\in(0,1)$, computes a graph $G_\epsilon$, which is a $(1\pm\epsilon)$-cut sparsifier for $G$ with high probability. The running time of the algorithm is $O(m\cdot\min(\alpha(n)\log (m/n), \log (n)))$ and the number of edges of $G_\epsilon$ is $O{\left(n\log (n)/\epsilon^2\right)}$. 
\end{theorem}

Using preprocessing with a result from \cite{FHHP19} (see Theorem~\ref{thm:preprocessing}), we obtain the following corollary for polynomially-weighted graphs.

\begin{corollary}
\label{cor:resultpoly}
There exists an algorithm that, given a polynomially-weighted graph $G$ and a freely chosen parameter $\epsilon\in(0,1)$, computes a graph $G_\epsilon$, which is a $(1\pm\epsilon)$-cut sparsifier for $G$ with high probability.
The running time of the algorithm is $O(m+n\left(\log^2(n)/\epsilon^2\right)\alpha(n) \log(\log (n)/\epsilon))$ and the number of edges of $G_\epsilon$ is $O(n\log (n)/\epsilon^2)$. 
\end{corollary}

Following Benczúr and Karger~\cite{BK15}, the computation of cut sparsifiers of graphs with fractional or even real weights can be reduced to integer weights. For the reduction see Appendix~\ref{app:reduction}. Thus our algorithm also gives a speedup for such graphs. Since the integer case is the essential one, we follow prior works and only formulate our results for this particular case.

\begin{figure}
    \centering
    \begin{tabular}{|l|l l|}
        \hline
        Algorithm & Size & Running time\\ \hline
        \textit{Unweighted} & & \\ %\hline 
        \cite{FHHP19} & $O\left(n\log (n)/\epsilon^2\right)$ & $O\left(m\right)$ \\ & & \\
        \textit{Polynomial weights} & & \\ %\hline
        \cite{BK15} & $O\left(n\log (n)/\epsilon^2\right)$ & $O\left(m \log^2 (n)\right)$\\
        \cite{FHHP19} & $O\left(n\log^2 (n)/\epsilon^2\right)$ & $O\left(m\right)$\\
        \cite{FHHP19} + \cite{BK15} & $O\left(n\log (n)/\epsilon^2\right)$ & $O\left(m+ n\log^4 (n)/\epsilon^2\right)$\\
        This paper & $O\left(n\log (n)/\epsilon^2\right)$ & $O(m\log (n))$\\
        This paper & $O\left(n\log (n)/\epsilon^2\right)$ & $O(m\alpha(n)\log(m/n))$\\
        \cite{FHHP19} + this paper & $O\left(n\log (n)/\epsilon^2\right)$ & $O\left(m+ n\left(\log^2 (n)/\epsilon^2\right)\alpha(n)\log(\log (n)/\epsilon)\right)$\\ & & \\
        \textit{Unbounded weights} & & \\ %\hline
        \cite{HP10} & $O\left(n\log^2 (n)/\epsilon^2\right)$ & $O\left(m \log^2 (n)/\epsilon^2\right)$\\
        \cite{BK15} & $O\left(n\log (n)/\epsilon^2\right)$ & $O\left(m \log^3 (n)\right)$\\
        \cite{HP10} + \cite{BK15} & $O\left(n\log (n)/\epsilon^2\right)$ & $O\left(m \log^2 (n)/\epsilon^2+n\log^5(n)/\epsilon^2\right)$\\
        \cite{LS17} & $O\left(n/\epsilon^2\right)$ & $O\left(m\cdot\text{poly}(\log (n),1/\epsilon)\right)$\\
        This paper & $O\left(n\log (n)/\epsilon^2\right)$ & $O(m\log (n))$\\
        This paper & $O\left(n\log (n)/\epsilon^2\right)$ & $O(m\alpha(n)\log(m/n))$\\ \hline
    \end{tabular}
    \caption{An overview of the state of the art algorithms for computing cut sparsifiers for undirected graphs with integer weights. Algorithm $A+B$ indicates that algorithm $B$ is preprocessed with algorithm $A$.}
    \label{fig:table}
\end{figure}

As a direct application of the cut sparsifier, we can use Theorem~\ref{thm:resultgeneral} and Corollary~\ref{cor:resultpoly} to replace $m$ by $n\log (n)/\epsilon^2$ in the time complexity of algorithms solving cut problems, at the cost of a $(1\pm\epsilon)$-approximation. We detail the effects for the minimum cut problem. Recently, Gawrychowski, Mozes, and Weiman \cite{GMW20} showed that one can compute the minimum cut of a weighted graph in $O(m\log ^2(n))$ time. Using sparsification \cite{BK15,FHHP19} for preprocessing, the state of the art for $(1+\epsilon)$-approximate min-cut is $O(m+n\log^4 (n)/\epsilon^2)$. When we use our new sparsification results, we obtain faster $(1+\epsilon)$-approximate min-cut algorithms when $m= \Omega(n\log(n)/\epsilon^2)$. 
\begin{corollary}
There exists an algorithm that, given a polynomially-weighted graph $G$ and a freely chosen parameter $\epsilon\in(0,1)$, with high probability computes an $(1+\epsilon)$-approximation of the minimum cut in time $O(m+n\log^3(n)/\epsilon^2)$. 

There exists an algorithm that, given a weighted graph $G$ and a freely chosen parameter $\epsilon\in(0,1)$, with high probability computes an $(1+\epsilon)$-approximation of the minimum cut in time $O(m\cdot\min(\alpha(n)\log (m/n),\log (n))+n\log^3(n)/\epsilon^2)$. 
\end{corollary}

For unweighted graphs, even faster minimum cut algorithms exist: Ghaffari, Nowicki, and Thorup~\cite{GNT20} show that we can find the minimum cut in $O(\min\{m+n\log^3(n),m\log(n)\})$ time. Combining this with the linear time cut sparsifier of Fung et al.~\cite{FHHP19}, we get $(1+\epsilon)$-approximate minimum cut in unweighted graphs in $O(m+n\log(n)\min\{1/\epsilon+\log^2(n),\log(n)/\epsilon\})$ time.  

The remainder of this article is organized as follows. The rest of the introduction consists of a technical overview of our algorithms. Section~\ref{sc:prelim} contains a review of the general sparsification framework from \cite{FHHP19} tailored to our needs, and can be skipped by readers that are already familiar with this work. We present our algorithm to compute the MSF indices in Section~\ref{sc:MSF}. This is used as a black box in our algorithm, which is presented and analyzed in Section~\ref{sc:ouralg}. In Section~\ref{sc:arbweights}, we show how the results of Section~\ref{sc:ouralg} generalize to graphs with unbounded weights.

\subsection*{Technical Overview}
\label{subsc:technicaloverview}
The high-level set-up of our sparsification algorithm is similar to the algorithm for unweighted graphs of Fung et al.~\cite{FHHP19}. Our main contribution consists of showing how to generalize this technique to weighted graphs, by using maximum spanning forest (MSF) indices instead of Nagamochi-Ibaraki (NI) indices. On a less significant note, we prove that by a tightening of the analysis one can show that the size and time bounds hold with high probability, and not only in expectation.

NI indices are defined by means of an NI forest packing: view graphs with integer weights as unweighted multigraphs, and repeatedly compute a spanning forest. The NI index is the (last) forest in which an edge appears (for details see Definition~\ref{def:NIindex}). The MSF index is also defined by a forest packing, but in this case the MSF packing: we say $\mathcal{F}=\{F_1, \dots, F_M\}$ is an \emph{$M$-partial maximum spanning forest packing} of $G$ if for all $i=1,\dots,M$, $F_i$ is a maximum spanning forest in $G\setminus \bigcup_{j=1}^{i-1}F_j$. Now, we say that an edge $e$ has MSF index $i$ (w.r.t.\ to some (partial) MSF packing  $\mathcal{F}$) if $e$ appears in the $i$-th forest $F_i$ of the (partial) MSF packing $\mathcal{F}$. The MSF index has been used previously in the context of dynamic graph sparsifiers (see Abraham et al.~\cite{Abraham2016OnFD}). However, there it was only used because it rendered a faster running time, but using NI indices in the corresponding static construction would have been possible as well. In this paper, we use distinctive properties of the MSF index, and the NI index would not suffice. We show that using the MSF index, we can generalize the sparsification algorithm for unweighted graphs to an algorithm for weighted graphs, thereby demonstrating that the MSF index is a natural analogue for the NI index in the weighted setting. We provide an algorithm to compute an $M$-partial MSF packing in time $O(m\cdot\min(\alpha(n)\log (M),\log (n)))$ for polynomially-weighted graphs. We show that for unbounded weights we can compute a sufficient estimation, also in time $O(m\cdot\min(\alpha(n)\log (M),\log (n)))$.

An important distinction between the unweighted algorithm of Fung et al.\ and our weighted algorithm, is that the use of contractions to keep running times low throughout the algorithm is no longer possible: edges of different weights have to be treated differently, hence cannot be contracted. By using multiple iterations with an exponentially decreasing precision parameter we can overcome this problem. 

In the case of a polynomially-weighted input graph, the algorithm consists of two main phases. In the first phase, we compute sets $F_0, F_1, \dots, F_\Gamma\subseteq E$, where edges satisfy some lower bound on the weight of any cut separating their endpoints. In the second phase, we sample edges from each set $F_i$ with a corresponding probability. 

We set a parameter $\rho =\Theta\left(\frac{\ln (n)}{\epsilon^2}\right)$ and start by computing a $2\rho$-partial maximum spanning forest packing for $G$. We define $F_0$ to be the union of these $2\rho$ forests. We add the edges of $F_0$ to $G_\epsilon$, which will become our sparsifier. We sample each of the remaining edges $E\setminus F_0$ with probability $1/2$ to construct $X_1$. To counterbalance for the sampling, we will boost the weight of each sampled edge with a factor 2. Now we continue along these lines, but in each iteration we let $F_i$ consist of an exponentially growing number of spanning forests: $F_i$ is defined as the union of the forests in a $(2^{i+1}\cdot\rho)$-partial MSF packing packing of $X_i$. Then, $X_{i+1}$ is sampled from the remaining edges $X_i\setminus F_i$, where again each edge is included with probability $1/2$. We continue this process until there are sufficiently few edges left in $X_{i+1}$. We add these remaining edges to $G_\epsilon$. 

The second phase of the algorithm is to sample edges from the sets $F_i$ and add these sampled edges to $G_\epsilon$. Hereto, note that an edge $e$ of $F_i$ (for $i\geq 1$) was not part of $F_{i-1}$, meaning it was not part of any spanning forest in a $(2^i\cdot \rho)$-partial MSF packing of $X_{i-1}$. This implies that for an edge $e\in F_i$ the weight of any cut $C$ in $X_{i-1}$ containing $e$ is at least $2^i\cdot \rho\cdot w(e)$. Now we use the general framework for cut sparsification of Fung et al.~\cite{FHHP19}, which boils down to the fact that this guarantee on the weights of cuts implies that we can sample edges from $F_i$ with probability proportional to $1/(2^iw(e))$. We show that this results in a sufficiently sparse graph. 

Intuitively, it might seem redundant to sample edges from $X_i\setminus F_i$ to form $X_{i+1}$. This is indeed not necessary to guarantee that the resulting graph is a sparsifier. However, it ensures that the number of iterations is limited, which leads to better bounds on the size of the sparsifier and the running time. Since we sample edges with probability $1/2$ in each phase, we need to repeat the sampling $O(\log(m/(m_0))$ times to get the size of $X_i$ down to $O(m_0)$. As this number of steps depends on the initial number of edges $m$, we get better bounds for size and running time if $m$ is already small. We will exploit this by preprocessing the graph with an algorithm from \cite{FHHP19} that gives a cut sparsifier of size $O(n\log^2(n)/\epsilon^2)$ in linear time. Moreover, we can show that repeatedly calling our algorithm has no worse asymptotic time bound than calling it once, since the input graph becomes sparser very quickly. By doing so, we obtain a sparsifier of size $O(n\log (n)/\epsilon^2)$. 

Since we only use that the MSF index gives a guaranteed lower bound on the connectivity of an edge, one might wonder why the NI index does not work here. After all, the NI indices of a graph can be computed in linear time, which would result in a significant speed-up. However, when computing the NI index, the weight of an edge influences the number of forests necessary, while computing the MSF index only requires the comparison of weights. Moreover, the number of trees in a MSF packing is always bounded by $n$. We can use this to bound the number of edges in the created sparsifier. The same technique with NI indices would make the size of the sparsifier depend on the maximum weight in the original graph.

To show that the algorithm outputs a cut sparsifier, it needs to be proven that both the sampling in the first and the second phase preserve cuts. We follow the lines of the analysis of \cite{FHHP19}, which makes use of cut projections and Chernoff bounds. We show that by partitioning the edge sets according to their weight this method extends to weighted graphs. 

One part of the algorithm has remained unaddressed: the computation of the maximum spanning forests. The approach we use here is related to Kruskal's algorithm for computing minimum spanning trees \cite{kruskal56}. We start by sketching the $M$-partial MSF packing algorithm for polynomial weights. We sort the edges according to their weights using radix sort in $O(m)$ time. We create $M$ empty forests on $n$ vertices. Starting with the heaviest edge, we add each edge $e$ to the first forest in which it does not create a cycle. We can find this forest using a binary search in $\log (M)$ steps. By using a disjoint-forest representation for the union-find data structure necessary to carry out these steps, we achieve a total time of $O(m\alpha(n)\log (M))$. 

When working with unbounded weights, the bottleneck is the initial sorting of the edges. Radix sort does not guarantee to be efficient for unbounded weights. Instead we could use a comparison-based algorithm, such as merge sort, which takes time $O(m\log (n))$. By employing a different data structure than before, we can guarantee total running time $O(m\log (n))$. However, we do not need the exact MSF indices for our sampling procedure, an estimate suffices. We can apply a `windowing' technique from \cite{BK15} to split the graph into subgraphs, where we can rescale the weights to polynomial weights and apply our previously mentioned algorithm. We then achieve a total running time of $O(m\alpha(n)\log (M))$, as before. For more details on this, we refer to Section~\ref{subsc:MSFA}. So in total we have running time $O(m\cdot\min(\alpha(n)\log(m/n),\log (n)))$. 

%In our sparsification algorithm, we will only need MSF indices up to $M=O(m/n)$, hence this encompasses good results with suitable preprocessing. As said, for polynomially-weighted graphs we use an algorithm from \cite{FHHP19} to compute a sparsifier of size $O(n\log^2(n)/\epsilon^2)$ in time $O(m)$. This implies that computing maximum spanning forests can be done in time $O(n\log^2(n)/\epsilon^2\alpha(n)\log(\log (n)/\epsilon))$ for the sparsified graph.

\section{Notation and Review}
\label{sc:prelim}
Throughout this paper, we consider $G=(V,E)$ to be an undirected, integer weighted graph on $|V|=n$ vertices with $|E|=m$ edges. We define a set of edges $C\subseteq E$ to be a \emph{cut} if there exists a partition of the vertices $V$ in two non-empty subsets $A$ and $B$, such that $C$ consists of all edges with one endpoint in $A$ and the other endpoint in $B$. The weight of the cut is the sum of the weights of the edges of the cut: $w_G(C)= \sum_{e\in C}w_G(e)$. The \emph{minimum cut} is defined as the cut with minimum weight. 
We say that a (reweighted) subgraph $H\subseteq G$ is a $(1\pm \epsilon)$-\emph{cut sparsifier} for a weighted graph $G$ if for every cut $C$ in $H$, its weight $w_H(C)$ is within a multiplicative factor of $1\pm\epsilon$ of its weight $w_G(C)$ in $G$. A key concept in the realm of cut sparsification is the connectivity of an edge. 
\begin{definition}
Let $G=(V,E)$ be a graph, possibly weighted. We define the \emph{connectivity} of an edge $e=(u,v)\in E$ to be the minimal weight of any cut separating $u$ and $v$. We say that $e$ is $k$-\emph{heavy} if it has connectivity at least $k$. For a cut $C$, we define the $k$-\emph{projection} of $C$ to be the $k$-heavy edges of the cut $C$. 
\end{definition}

The following theorem from \cite{FHHP19} bounds the number of distinct $k$-projections of a graph, it is a generalization of a preceding theorem by Karger, see \cite{karger93,karger96}. This result can be useful when showing that cuts are preserved by a sampling scheme. This is due to the fact that while there may be exponentially many different cuts, this theorem shows that there are only polynomially many cut projections. Hence if one can reduce a claim for cuts to their $k$-projections, a high probability bound can be obtained through the application of a Chernoff bound.

\begin{theorem}%[See {\cite[Theorem 2.3]{FHHP19}}]
\label{thm:cutprojection}
For any $k\geq \lambda$ and any $\eta\geq 1$, the number of distinct $k$-projections in cuts of weight at most $\eta k$ in a graph $G$ is at most $n^{2\eta}$, where $\lambda$ is the weight of a minimum cut in $G$.
\end{theorem}

Throughout this paper, we say a statement holds \emph{with high probability} (w.h.p.) if it holds with probability at least $1-n^c$, for some constant $c$. This constant can be modified by adjusting the constants hidden in asymptotic notation.

\subsection{A General Framework for Cut Sparsification} 
\label{subsc:framework}
We review the general framework for cut sparsification as presented in \cite{FHHP19}. This section does not contain new results, and can be skipped by readers that are only interested in our contribution. 

The framework shows that edges can be sampled using different notions of connectivity estimators. Although this scheme provides one proof for the validity of multiple parameters, it might be worth noting that an analysis tailored to the used connectivity estimator might provide a better result. For example, when the framework is applied with `edge strengths', it produces a sparsifier of size $O(n\log^2(n)/\epsilon^2)$, a $\log (n)$ factor denser than the edge strength-based sparsifier of Benczúr and Karger \cite{BK15}. 

Let $G=(V,E)$ be a graph with integer weights, and let $\epsilon \in (0,1)$, $c\geq 1$ be parameters. Given a parameter $\gamma$ (possibly depending on $n$) and an integer-valued parameter $\lambda_e$ for each $e\in E$. We obtain $G_\epsilon$ from $G$ by independently \emph{compressing} each edge $e$ with parameter 
\[ p_e = \min\left(1,\frac{16(c+7)\gamma \ln (n)}{0.38 \lambda_e \epsilon^2}\right).\] %can be 16(c+6)
Compressing an edge $e$ with weight $w(e)$ consists of sampling $r_e$ from a binomial distribution with parameters $w(e)$ and $p_e$. If $r_e>0$, we include the edge in $G_\epsilon$ with weight $r_e/p_e$.

In the following we describe a sufficient condition on the parameters $\gamma$ and $\lambda_e$ such that $G_\epsilon$ is a $(1\pm\epsilon)$-cut sparsifier for $G$ with probability at least $1-4/n^c$. Hereto we partition the edges according to their value $\lambda_e$:
\begin{align*}
    \Lambda &:= \left\lfloor\log \left(\max_{e\in E}\{\lambda_e\}\right)\right\rfloor; \\
    R_i &:= \{ e\in E : 2^i\leq \lambda_e\leq 2^{i+1}-1\}.
\end{align*}
Let $\mathcal G=\{G_i=(V,E_i) : 1\leq i\leq \Lambda\}$ be a set of integer-weighted subgraphs such that $R_i\subseteq G_i$. Moreover suppose that $w_{G_i}(e)\geq w_G(e)$ for each $e\in R_i$. For a given set of parameters $\Pi=\{\pi_1, \dots, \pi_\Lambda\}\subseteq \mathbb R^\Lambda$, we define
\begin{itemize}
    \item $\Pi$-\textit{connectivity}: each edge $e\in R_i$ is $\pi_i$-heavy in $G_i$;
    \item $\gamma$-\textit{overlap}: for any cut $C$,
    \[ \sum_{i=0}^\Lambda \frac{e_i^{(C)}2^{i-1}}{\pi_i}\leq \gamma \cdot e^{(C)},\]
    where $e^{(C)}=\sum_{e\in C}w_G(e)$ and $e_i^{(C)}=\sum_{e\in C\cap E_i}w_{G_i}(e)$.
\end{itemize}

The following theorem shows that compressing with parameters adhering to these conditions gives a cut sparsifier with high probability. 

\begin{theorem}[See {\cite[Theorem 1.14]{FHHP19}}]
\label{thm:framework}
Fix the parameters $\gamma$ and $\lambda_e$ for each edge $e$. If there exists $\mathcal G$ satisfying $\Pi$-connectivity and $\gamma$-overlap for some $\Pi$, then $G_\epsilon$ is a $(1\pm \epsilon)$-cut sparsifier for $G$, with probability at least $1-4/n^c$, where $G_\epsilon$ is obtained by edge compression using parameters $\gamma$ and $\lambda_e$'s. 
%Furthermore, $G_\epsilon$ has $O{\left(\frac{\gamma \log (n)}{\epsilon^2}\sum_{e\in E}\frac{1}{\lambda_e}\right)}$ edges in expectation. 
\end{theorem}

\subsection{A First Application of the Framework}
In this section, we review the application of the framework from the previous section with \emph{Nagamochi-Ibaraki (NI) indices} as parameters, as presented in \cite{FHHP19}. As the name suggests, NI indices were first introduced by Nagamochi and Ibaraki \cite{NI92computing,NI92linear}. The algorithm they provide gives a graph partitioning into forests, and subsequently a corresponding index for each edge, called the NI index. 

\begin{definition}\label{def:NIindex}
Let $G=(V,E)$ be a graph, possibly weighted. We say an edge-disjoint sequence $F_1, F_2, \dots$ of forests is a \emph{Nagamochi-Ibaraki forest packing} for $G$ if $F_i$ is a spanning forest for $G\setminus \bigcup_{j=1}^{i-1}F_i$, where the weights of $\bigcup_{j=1}^{i-1}F_i$ are subtracted of $G$. If $G$ is a weighted graph, each edge $e$ must be contained in $w(e)$ contiguous forests. We define the \emph{NI index}, denoted by $l_e$, to be the index of the (last if weighted) forest in which $e$ appears. 
\end{definition}

Nagamochi and Ibaraki show that the NI indices can be computed in linear time for unweighted graphs and in $O(m+n\log (n))$ time for weighted graphs, see \cite{NI92linear,NI92computing}. As is shown in \cite{FHHP19}, we can use the NI index as the connectivity estimator in the sparsification framework to obtain the following result.

\begin{theorem}%[See {\cite[Theorem 1.20]{FHHP19}}]
\label{thm:preprocessing}
Let $G=(V,E)$ be a weighted graph, and let $\epsilon>0$ be a constant. Let $G_\epsilon$ be obtained by independently compressing each edge with parameter $p_e = \min(1,\rho/l_e)$, where $\rho = \frac{224}{0.38}\ln (n)/\epsilon^2$. Then $G_\epsilon$ is a $(1\pm\epsilon)$-cut sparsifier for $G$ with high probability. 
\end{theorem}
The sampling itself takes at most $O(m)$ time, as explained in Section \ref{subsc:time}. As the NI indices can be computed in $O(m+n\log (n))$ time, this implies that the total running time is $O(m+n\log (n))$. As a graph with $m\leq n\log (n)$ is already sparse, we can assume $m> n\log (n)$. Thus, for our purposes, the total running time is simply $O(m)$.

Next we provide a bound for the number of edges in the sparsifier $G_\epsilon$. \cite{FHHP19} proves this same bound in expectation, we provide a proof for this bound `with high probability'. 

\begin{lemma}\label{lm:preprocessingmaxdeg}
With high probability, the size of the graph $G_\epsilon$ in Theorem~\ref{thm:preprocessing} is $O(n\log^2(n)/\epsilon^2)$.
\end{lemma}
\begin{proof}
Let $v\in V$ be a vertex with degree $d_v\geq O(\log^2 (n)/\epsilon^2)$ in $G$. We denote the degree of $v$ in $G_\epsilon$ by $d_v'$ and we write $d':= \max_{v\in V} d_v'$. For each neighbor $u$ of $v$ in $G$, we compress the edge $e=(u,v)$ with parameter $p_e =\min\left(1, \frac{224\ln (n)}{0.38\epsilon^2l_e}\right)$, where $l_e$ is the NI index of $e$. For each edge, the probability that it remains after compression is $1-(1-p_e)^{w_e}$. From Bernoulli's inequality we see $1-(1-p_e)^{w_e}\leq w_ep_e$. Let $Y_e$ be the random variable that is $1$ if $e$ remains, and $0$ else. We note that $\E\left[ \sum_{e : v\in e}Y_e\right]\leq \frac{224}{0.38}\ln^2(n)/\epsilon^2$. 
Now we apply a Chernoff bound (Theorem~\ref{thm:Chernoffupperbound}) to obtain
\begin{align*}
	\P\left[ d_v' \geq \delta \frac{224}{0.38}\ln^2(n)/\epsilon^2\right] &\leq \exp\left(-0.38\delta \frac{224}{0.38}\ln^2(n)/\epsilon^2\right) = n^{-224\delta \ln (n)/\epsilon^2}.
\end{align*}
Using a union bound we get the desired result
\[ \P\left[ d' \leq \delta \frac{224}{0.38}\ln^2(n)/\epsilon^2\right] \geq 1- n^{1-224\delta \ln (n)/\epsilon^2}.\]
Consequently, we obtain that with high probability the number of edges of the sparsifier is at most $O(n\log^2(n)/\epsilon^2)$. 
\end{proof}

The state of the art for polynomially-weighted graphs is achieved by postprocessing this result with the algorithm by Benczúr and Karger \cite{BK15}. Thus our improvement on \cite{BK15} leads to an overall improved result. 

\subsection{The Computational Model}\label{sc:model}
If we have an input graph $G=(V,E)$ with weights $w\colon E \to \{1, \dots, W\}$, we assume our computational model has word size $\Theta(\log (W)+\log (n))$. Note that for polynomial weights, this comes down to a word size of $\Theta(\log (n))$. Moreover, we assume that basic operations on such words have uniform cost, i.e., they can be performed in constant time. In particular, these basic operations are addition, multiplication, inversion, logarithm, and sampling a random bit string of word size precision. Such assumptions are in line with previous work \cite{BK15,FHHP19}, where they are made implicitly. 

\section{A Maximum Spanning Forest Packing}\label{sc:MSF}
An important primitive in our algorithm is the use of the maximum spanning forest (MSF) index. The concept is similar to the Nagamochi-Ibaraki index, the important difference is that an edge $e$ with weight $w(e)$ appears in $w(e)$ different NI forests. This means that the number of NI forests depends on the numerical values of the edge weights, and thus can grow far beyond $O(n)$. On the other hand, the number of maximum spanning forests in a MSF packing is bounded by the maximum degree in the graph, hence also by $n$. While this already has noteworthy implications for polynomially-weighted graphs, it is even more significant for superpolynomially-weighted graphs. We believe that this property might make them suitable for applications other than presented here. 

\begin{definition}\label{def:MSF}
Let $G=(V,E)$ be a weighted graph. We say $\mathcal{F}=\{F_1, \dots, F_M\}$ is an \emph{$M$-partial maximum spanning forest packing} of $G$ if for all $i=1,\dots,M$, $F_i$ is a maximum spanning forest in $G\setminus \bigcup_{j=1}^{i-1}F_j$. If we have that $\bigcup_{i=1}^M F_i = G$, then we call $\mathcal{F}$ a \emph{(complete) maximum spanning forest packing} of $G$. Moreover, for $e\in E$ we denote the \emph{MSF index} of $e$ (w.r.t.\ $\mathcal{F}$) by $f_e$, i.e., $f_e$ is the unique index such that $e\in F_{f_e}$. 
\end{definition}

Note that we do not demand the $F_i\in \mathcal{F}$ to be non-empty, as this suits notation bests in our applications. Also note that a (partial) MSF packing is fully determined by the MSF indices. 

The following theorem states that computing the MSF indices up to $M$ takes $O(m\alpha(n)\log (M))$ time for polynomially-weighted graphs. 

\begin{theorem}
\label{thm:MSFP}
Let $G=(V,E)$ be a graph, where we allow parallel edges but no self-loops, and we suppose $m \leq n^2$. Suppose we have weights $w\colon E \to \{1,\dots, n^c\}$ for some $c\geq 0$. Then, for any $M>0$, there exists an algorithm that computes an $M$-partial MSF packing in $O(m(\alpha(n)\log (M)+c))$ time. 
\end{theorem}
\begin{proof}
The outline of the algorithm is as follows. 
\begin{enumerate}
\item Sort the edges by weight in descending order using radix sort in base $n$.\footnote{Note that conversion to base $n$ takes time $O(\log_n (w(e))) \leq O(\log_n (n^c)) = O(c)$ for each edge, so total time $O(mc)$.} \label{step:sort}
\item Create empty forests $ F_1, \ldots, F_M $.\label{step:initialize forests}
\item Iterate over the edges in descending order and for each edge $ e = (u, v) $ do the following: \label{step:iterate}
\begin{enumerate}
\item Find the smallest index $ i $ such that $ u $ and $ v $ are not connected in $ F_i $. \label{step:find smallest index}\label{step:u and v connected in F_i}
\item Store $ i $ as the MSF index $f_e$ of $ e $. If $u$ and $v$ are connected in every $F_i$, store $f_e>M$. 
\item Add $ e $ to $ F_i $.\label{step:add edge to F_i}
\end{enumerate}
\end{enumerate}

We need at most $M$ trees, since we only compute an $M$-partial MSF packing. By using radix sort, the initial sorting takes time $ O (cm) $ time (for a time bound of radix sort, see e.g.\ \cite{CLRS}). We show that the remainder of the algorithm can be executed in $O(m\alpha(n)\log (M))$ time. 

For every $ 1 \leq i \leq M $ we maintain the non-singular components of $ F_i $ with a union-find data structure (supporting the three operations $ \textsc{MakeSet}_i $, $ \textsc{Union}_i $, and $ \textsc{FindSet}_i $). To be precise, we use the disjoint-set forest representation of Tarjan \cite{tarjan75} (see e.g.\ {\cite[Chapter 21]{CLRS}}). Additionally, for every node $ v \in V $ we maintain $ s (v) $, the smallest index $ i $ such that $ \{ v \} $ is a singleton component of $ F_i $.

For the binary search in Step~\ref{step:find smallest index} it is sufficient to first search over indices $ i < \min \{ s (u), s (v) \} $. If this search is successful and we find such an index $ i < \min \{ s (u), s (v) \} $, then we perform $ \textsc{Union}_i (u, v) $. Otherwise, we have learned that $ \min \{ s (u), s (v) \} $ is the smallest index $ i $ such that $ u $ and $ v $ are not connected in $ F_i $.
The algorithm then proceeds as follows:
\begin{itemize}
\item Let $ j := \min \{ s (u), s (v) \} $.
\item If $ j = s (u) $, then we perform $ \textsc{MakeSet}_j (u) $ and increase $ s (u) $ by one.
\item If $ j = s (v) $ (which could also be the case in addition to $ j = s (u) $), we perform $ \textsc{MakeSet}_j (v) $ and increase $ s (j) $ by one.
\item Finally, we perform $ \textsc{Union}_j (u, v) $.
\end{itemize}

Now let $ \varphi_i $, $\chi_i$, and $ \psi_i $ denote the number of $\textsc{MakeSet}_i$-, $\textsc{Union}_i$-, and $\textsc{FindSet}_i$-operations in the $i$-th union-find data structure, respectively. Since we use the disjoint-set forest representation, we obtain a bound on the running time (see {\cite[Theorem 21.14]{CLRS}}) of $ O ((\varphi_i + \chi_i+\psi_i)\alpha(\varphi_i + \chi_i+\psi_i)) $. To obtain a bound on the total running time for all operations in the union-find data structures, we sum over all $i$:
\begin{align*} 
O{\left(\sum_{i=1}^M (\varphi_i + \chi_i+\psi_i)\alpha(\varphi_i + \chi_i+\psi_i)\right)} \leq O{\left(\sum_{i=1}^M (\varphi_i + \chi_i+\psi_i)  \alpha{\left(\max_{j=1,\dots M} (\varphi_j + \chi_j+\psi_j)\right)}\right)}. 
\end{align*}

Now observe that in total we perform at most two \textsc{MakeSet}-operations per edge (one for each of its endpoints) and thus $ \sum_{i=1}^M \varphi_i \leq 2 m $. The number of \textsc{Union}-operations is always bounded above by $ \varphi_i - 1$, so $ \sum_{i=1}^M \chi_i \leq 2 m $. Furthermore, by using binary search, we perform $ O (\log (M)) $ \textsc{FindSet}-operations per edge and thus $ \sum_{i=1}^M \psi_i = O (m \log (M)) $.
Thus we get a total time of $O((m\log (M)) \alpha(m+2n))\leq O{\left(m\alpha(n)\log (M)\right)}$, which holds because $\alpha{\left(m+2n\right)}\leq \alpha{\left(n^4\right)}\leq \alpha(n)+2$. 

We therefore arrive at a total running time of $ O{\left(cm +m\alpha(n)\log (M)\right)} =O{\left(m \left(\alpha(n)\log (M)+c\right)\right)}$.  
\end{proof}

If we want to compute the full maximum spanning forest packing, it suffices to set $M$ to be the maximum degree in the graph. When $M$ is large, managing the data structures slightly differently yields a better result.

\begin{theorem}
\label{thm:MSFdense}
There exists an algorithm that, given a weighted graph $G=(V,E)$, where we allow parallel edges but no self-loops, and parameter $M>0$, computes an $M$-partial MSF packing in $O(m(\log (n)+\log (M)))$ time. 
\end{theorem}
\begin{proof}
We use the same algorithm as in Theorem~\ref{thm:MSFP} with two simple changes. In step~\ref{step:sort}, we use an optimal comparison-based sorting algorithm, like merge sort, instead of radix sort. This takes time $O(m\log(n))$.
In steps~\ref{step:initialize forests} and \ref{step:iterate}, we use a linked-list representation \cite[Chapter 21]{CLRS} instead of the disjoint-set forest representation. To analyze the running time, recall the following notation. Let $ \varphi_i $, $\chi_i$, and $ \psi_i $ denote the number of $\textsc{MakeSet}_i$-, $\textsc{Union}_i$-, and $\textsc{FindSet}_i$-operations in the $i$-th union-find data structure, respectively. Here $\textsc{MakeSet}_i$-, $\textsc{Union}_i$-, and $\textsc{FindSet}_i$ are the operations on the component $F_i$. By \cite[Theorem 21.1]{CLRS}, we obtain a bound on the running time of: $O(\chi_i + \psi_i + \varphi_i\log (\varphi_i))$. We sum over all $i$ to obtain
\begin{align*}
O\left( \sum_{i=1}^M \chi_i + \psi_i + \varphi_i\log(\varphi_i)\right) = O\left( \sum_{i=1}^M \chi_i + \sum_{i=1}^M + \psi_i +  \sum_{i=1}^M  \varphi_i\log (\varphi_i)\right).
\end{align*}
As before, we have $\sum_{i=1}^M \chi_i \leq m$ and $\sum_{i=1}^M \varphi_i \leq 2m$. Also note $\varphi_i \leq n$, as $F_i$ is a forest, so we have $\sum_{i=1}^M \varphi_i\log (\varphi_i) \leq 2m\log (n)$. Again, we perform at most $O(\log (M))$ $\textsc{FindSet}$-operations per edge, hence $\sum_{i=1}^M \psi_i = O(m\log (M))$. We conclude we have total time $O(m(\log (n) + \log (M))$. 
\end{proof}

Note that if we do not have parallel edges, then $M\leq n$, so the running time simplifies to $O(m\log (n))$. Also note that the weights no longer need to be bounded for this result. In the next section, we consider an algorithm for \emph{sparse} graphs with unbounded weights.

\subsection{An Estimation for Unbounded Weights}\label{subsc:MSFA}
For our purposes we do not need the exact MSF indices, but an estimate suffices. The MSF index guarantees that if an edge $e=(u,v)\in E$ has MSF index $f_e$, then there are at least $f_e$ paths from $u$ to $v$, where every edge on such a path has weight at least $w(e)$. We relax this, to get the guarantee that if an edge $e=(u,v)\in E$ has estimated MSF index $\tilde{f}_e$, then there are at least $\tilde{f}_e$ paths from $u$ to $v$, where every edge on such a path has weight at least $(1-1/n)w(e)$. When we only compute estimates, we can do this faster than when we compute exact indices. The following lemma is inspired by the windowing technique of Benczúr and Karger \cite{BK15}, which shows that strong connectivities can be computed efficiently for graphs with unbounded weights by `windowing' these weights. This means we divide the graph into subgraphs according to an estimate and compute the sought connectivity estimators in these subgraphs. Hereto, we first compute a single maximum spanning forest $F$ for $G$. Now we define $d(e)$ to be the minimum weight among the edges on the path from $u$ to $v$ in $F$, where $e=(u,v)$. This can be done in total time $O(m+n)$, see \cite{Thorup99}.

\begin{lemma}\label{lm:MSFA}
There exists an algorithm that, given a weighted graph $G=(V,E)$ and parameter $M>0$, computes in time $O(m\alpha(n)\log (M))$ an MSF index estimator $\tilde{f}_e$ for each edge $e\in E':=\{e\in E : w(e)> d(e)/n\}$ with $f_e\leq M$.
\end{lemma}
\begin{proof}
We will split the graph $G$ into graphs $G^{(D)}$ for different values of $D$. In each $G^{(D)}$ we compute the estimator $\tilde{f}_e$ for some subset of edges from $E'$. We iteratively define $D$ to be the highest value among the $d(e)$ for which $e\in E'$ and $\tilde{f}_e$ has not been computed yet. We look at the subgraph $G^{(D)}=(V_D,E_D)$ defined by contracting all edges with $w(e)>D$, and deleting self loops. Moreover, we delete all edges with $w(e)\leq D/n^2$.

We claim we that for each edge $e\in E'$ with $d(e) \in (D/n,D]$ the MSF index in $G$ is equal to the MSF index in $G^{(D)}$. First we show that these edges actually appear in $G^{(D)}$. It is clear that $w(e) \leq d(e)\leq D$, so such an edge $e$ is not contracted. Suppose $e$ is deleted, then $w(e) \leq D/n$. If $d(e) \in (D/n,D]$, then $d(e) \geq D/n$, hence $D/n^2 \leq d(e)/2$. Consequently $w(e) \leq D/n \leq d(e)/n$, so then $e\notin E'$. 

Now that we have established that edges $e\in E'$ with $d(e) \in (D/n,D]$ appear in $G^{(D)}$, it remains to show that we can compute $\tilde{f}_E$. First let us remark that if $e=(u,v)$, then no (relevant) path from $u$ to $v$ is eliminated, since each such path must have an edge $e'$ with $w(e')\leq d(e)\leq D$, by definition of $d(e)$. Hence the only paths that are deleted, contain an edge $e'$ with $w(e') \leq D/n^2 < d(e)/n \leq w(e)$, hence this path does not contribute to the MSF index. 

Next we compute an estimator of the MSF index in $G^{(D)}$, by computing the MSF indices in a reweighted graph. We rescale the graph by multiplying all weights with $n^3/D$ and rounding to the closest integer. This means that we have an error in the weight of at most $D/n^3$. For an edge with $D/n^2<w(e)\leq D$, this means that the error is at most $w(e)/n$. So using Theorem~\ref{thm:MSFP}, we can compute the MSF indices in this reweighted sugraph with edge weights bounded by $n^3$ in time $O(m'\alpha(n)\log (M))$, where there is a multiplicative error in edge weights of at most $(1\pm 1/n)$. Note that each edge appears in at most two subgraphs, hence we have a total time of $O(m \alpha(n) \log (M))$.
\end{proof}

\section{Cut Sparsification for Weighted Graphs}
\label{sc:ouralg}
In this section, we present our algorithm for computing a $(1\pm \epsilon)$-cut sparsifier $G_\epsilon$ for a weighted graph $G$. This makes use of the framework as presented in Section~\ref{subsc:framework} and the maximum spanning forest packing as treated in Section~\ref{sc:MSF}.
This section works towards proving the following theorem for polynomially-weighted graphs. In Section~\ref{sc:arbweights}, we will generalize the techniques of this section to graphs with unbounded weights.

\begin{theorem}
\label{thm:ourcontribution}
There exists an algorithm that, given a weighted graph $G=(V,E)$, and freely chosen parameter $\epsilon>0$, computes a graph $G_\epsilon$, which is a $(1\pm\epsilon)$-cut sparsifier for $G$ with high probability. The algorithm runs in time $O(m\cdot\min(\alpha(n)\log (m/n),\log (n)))$ and the number of edges of $G_\epsilon$ is $O{\left(n\left(\log (n)/\epsilon^2\right)\log\left(m/(n\log (n)/\epsilon^2)\right)\right)}$. 
\end{theorem}
To be precise, we give an algorithm where the given bounds on both running time and size of the sparsifier hold with high probability. By simply halting when the running time exceeds the bound, and outputting an empty graph if we exceed the size bound, this gives the result above. 

To achieve a better bound on the size of the sparsifier, we repeatedly apply this theorem to the input graph, with an exponentially decreasing precision parameter.

\begingroup
\def\thetheorem{\ref{thm:resultgeneral}}
\begin{theorem}[Restated]
%\label{thm:impr_main}
There exists an algorithm that, given a weighted graph $G=(V,E)$, and freely chosen parameter $\epsilon\in(0,1)$, computes a graph $G_\epsilon$, which is a $(1\pm\epsilon)$-cut sparsifier for $G$ with high probability. The algorithm runs in time $O(m\cdot\min(\alpha(n)\log (m/n),\log (n)))$ and the number of edges of $G_\epsilon$ is $O{\left(n\log (n)/\epsilon^2\right)}$. 
\end{theorem}
\addtocounter{theorem}{-1}
\endgroup
\begin{proof}We obtain this result by repeatedly applying the algorithm from Theorem~\ref{thm:ourcontribution}, for a total of $k:=\log^* \left(\frac{m}{n\log(n)/\epsilon^2}\right)$ times. In iteration $i$, we set $\epsilon_i := \epsilon/2^{k-i+2}$ and denote the output of this iteration by $G_i$. This means that $G_i$ is a $(1\pm \epsilon/2^{k-i+2})$-cut sparsifier for $G_{i-1}$. In total, we see that $G_\epsilon:=G_k$ is a $(1\pm\epsilon)$-cut sparsifier for $G$ since
\begin{align*}
\prod_{i=1}^k (1+\epsilon/2^{k-i+2}) &\leq \exp\left(\sum_{i-1}^k\log(1+\epsilon/2^{k-i+2})\right)   \leq  \exp\left(\sum_{i-1}^k\epsilon/2^{k-i+2}\right)\\ 
&\leq \exp\left(\epsilon \sum_{j=2}^\infty 2^{-j}\right) = \exp(\epsilon/2) \leq 1+\epsilon,
\end{align*}
as $\epsilon<1$, and 
\begin{align*}
\prod_{i=1}^k (1-\epsilon/2^{k-i+2}) &\geq \prod_{j=0}^{\infty} 1- \frac{\epsilon/4}{2^{j}}  &&\geq (1-\epsilon/8)\prod_{j=1}^{\infty} 1-\frac{\epsilon/4}{j^2} \\
&=  (1-\epsilon/8)\frac{\sin(\pi\sqrt{\epsilon}/2)}{\sqrt{\epsilon}/2} &&\geq  (1-\epsilon/8)(1- \pi^2/24\epsilon) \\
&\geq 1-(1/8+\pi^2/24)\epsilon+\frac{\pi^2}{192}\epsilon^2 &&\geq 1-\epsilon.
\end{align*}
Since $k=\log^* \left(\frac{m}{n\log(n)/\epsilon^2}\right)=O(\log^* (n))$, all bounds hold with high probability simultaneously, and thus the end result holds with high probability. 

Now for the size bound, we have that $$m_i := |E(G_i)| \leq C\cdot\left(\frac{n\log(n)}{\epsilon^2} 4^{k-i+2} \log\left(\frac{m_{i-1}}{n\log(n)/\epsilon^2}\right)\right),$$
for some constant $C>0$, where we denote $m_0:=m$. We will show by induction that 
$$m_i \leq C\cdot\left(\frac{n\log(n)}{\epsilon^2} 4^{k-i+2}\cdot 2 \log^{(i)}\left(\frac{m}{n\log(n)/\epsilon^2}\right)\right),$$
which means in particular that $m_k = O{\left(n\log (n)/\epsilon^2\right)}$. The claim for $m_1$ is immediate. Suppose it holds for $i-1$, then
\begin{align*}
	m_{i} &\leq C\cdot\left(\frac{n\log(n)}{\epsilon^2} 4^{k-i+2} \log\left(\frac{m_{i-1}}{n\log(n)/\epsilon^2}\right)\right)\\
&\leq C\cdot\left(\frac{n\log(n)}{\epsilon^2} 4^{k-i+2} \log\left(C\cdot4^{k-i+3} \cdot 2\log^{(i-1)}\left(\frac{m}{n\log(n)/\epsilon^2}\right)\right)\right)\\
&= C\cdot\left(\frac{n\log(n)}{\epsilon^2} 4^{k-i+2} \left((k-i)\log(4)+\log(C\cdot2^7)+\log^{(i)}\left(\frac{m}{n\log(n)/\epsilon^2}\right)\right)\right)\\
&\leq C\cdot\left(\frac{n\log(n)}{\epsilon^2} 4^{k-i+2} \cdot 2\log^{(i)}\left(\frac{m}{n\log(n)/\epsilon^2}\right)\right),
\end{align*}
since 
\begin{align*}
(k-i)\log(4)+\log(C\cdot2^7)&= \left(\log^* \left(\frac{m}{n\log(n)/\epsilon^2}\right)-i\right)\log(4)+\log(C\cdot2^7)\\
&=\log^*\left(\log^{(i)}\left(\frac{m}{n\log(n)/\epsilon^2}\right)\right)\log(4)+\log(C\cdot2^7)\\
&< \log^{(i)}\left(\frac{m}{n\log(n)/\epsilon^2}\right), 
\end{align*}
if $\frac{m}{n\log(n)/\epsilon^2}>D$, for some constant $D$. This can be assumed to hold, since if $\frac{m}{n\log(n)/\epsilon^2}\leq D$, then Theorem~\ref{thm:ourcontribution} immediately gives the desired result. 
The total running time becomes of the sum of the $k$ iterations:
$$\sum_{i=1}^k O(m_{i-1}\cdot\min(\alpha(n)\log (m_{i-1}/n),\log (n))) = O\left(\left(m+\sum_{i=1}^{k-1} m_i\right)\cdot\min(\alpha(n)\log (m_{i-1}/n),\log (n)))\right).$$
Note that 
\begin{align*}
	\sum_{i=1}^{k-1} m_i &\leq \sum_{i=1}^{k-1}C\cdot\left(\frac{n\log(n)}{\epsilon^2} 4^{k-i+2}\cdot 2 \log^{(i)}\left(\frac{m}{n\log(n)/\epsilon^2}\right)\right)\\
&= O\left(\frac{n\log(n)}{\epsilon^2} 4^{k} \log\left(\frac{m}{n\log(n)/\epsilon^2}\right) \right). 
\end{align*}
We have $\log^*(x) = O(\log\log(x))$, hence we obtain $4^{\log^*(x)}\log(x) = O(\log^2(x))=O(x)$. Using this with $x = \frac{m}{n\log(n)/\epsilon^2}$ gives us total running time 
\begin{align*}
&\sum_{i=1}^k O(m_{i-1}\cdot\min(\alpha(n)\log (m_{i-1}/n),\log (n)))\\
= &O\left(\left(m+\frac{n\log(n)}{\epsilon^2} 4^{k} \log\left(\frac{m}{n\log(n)/\epsilon^2}\right)\right)\cdot\min(\alpha(n)\log (m/n),\log (n)) \right) \\
= &O(m\cdot\min(\alpha(n)\log (m/n),\log (n))).\qedhere
\end{align*}
\end{proof}

\subsection{The Algorithm}
To sparsify the graph, two methods of sampling are used. One of which is the framework presented in Section~\ref{subsc:framework}. However, instead of applying the framework to the graph directly, there is another sampling process that precedes it. 

To simplify equations, let us set $\rho :=\frac{(7+c)1352 \ln (n)}{0.38 \epsilon^2}$. If $|E|\leq 4\rho n\log\left(m/(n\log (n)/\epsilon^2)\right)$, we do nothing. That is, we return $G_\epsilon=G$. If not, we start by an initialization step and continue with an iterative process, which ends when the remaining graph becomes sufficiently small.

In the initialization step, we define $X_0 := E$. We compute an $\lfloor 2\rho\rfloor$-partial maximum spanning forest packing $T_1, \dots,T_{\lfloor 2\rho\rfloor}$ and we define $F_0 := \bigcup_{j=1}^{\lfloor 2\rho\rfloor}T_j$. The remaining edges $Y_0:= X_0\setminus F_0$ move on to the next phase. 

In iteration $i$, we create $X_{i+1}$ from $Y_i$ by sampling each edge with probability $1/2$. Next, we compute $k_i:=\rho\cdot 2^{i+1}$ maximum spanning forests $T_1, \dots, T_{k_i}$. We define $F_i := \bigcup_{j=1}^{k_i} T_j$, and $Y_i:= X_i \setminus F_i$.  

We continue until $Y_i$ has at most $2\rho n$ edges, and set $\Gamma$ to be the number of iterations. We retain all edges in $F_0$. In other words: add each edge $e\in F_0$ to $G_\epsilon$ with weight $w(e)$. The edges of $Y_\Gamma$ are also retained, but they need to be scaled to counterbalance the $\Gamma-1$ sampling steps: add each edge $e\in Y_\Gamma$ to $G_\epsilon$ with weight $2^{\Gamma-1}w(e)$. 

Any other edge $e\in F_i$ is at least $k_iw(e)$-heavy in $X_{i-1}$, as $e\notin F_{i-1}$. We exploit this heavyness to sample from these edges using the framework. For each $e\in F_i$ we:
\begin{itemize}
\item Define $n_e:=2^iw(e)$ and $p_e:= \min\left( 1,\frac{384}{169}\frac{1}{4^iw(e)}\right)$;
\item Generate $r_e$ from the binomial distribution with parameters $n_e$ and $p_e$;
\item If $r_e$ is positive, add $e$ to $G_\epsilon$ with weight $r_e/p_e$.
\end{itemize}
The factor $2^i$ in calling upon the binomial distribution can be seen as boosting the weight of the edge by a factor $2^i$, which is needed to counterbalance the $i$ sampling steps in creating $F_i$. 

Up to the computation method of the MSF packing, the presented algorithm is the same for polynomially and superpolynomially-weighted graphs. For the unbounded case, we use the MSF index estimator as presented in Section~\ref{subsc:MSFA}. In Section~\ref{sc:arbweights} we detail how this influences the correctness of the algorithm, and the bounds on size and running time. 

\vspace{1em}
\begin{algorithm}[hbt!]
\SetAlgoLined \caption{\textsc{Sparsify}$(V,E,w,\epsilon,c)$}\label{alg:mainalg}
\KwIn{An undirected graph $G=(V,E)$, with integer weights $w\colon E\to \N^+$, and parameters $\epsilon\in(0,1)$, $c\geq 1$.}
\KwOut{An undirected weighted graph $G_\epsilon=(V,E_\epsilon$).}
Set $\rho\leftarrow \frac{(7+c)1352 \ln (n)}{0.38 \epsilon^2}$.\\
\If{$|E|\leq 4\rho n\log\left(m/(n\log (n)/\epsilon^2)\right)$}{\Return{$G_\epsilon=G$.}} 
Compute an $\lfloor 2\rho \rfloor$-partial maximum spanning forest packing $T_1,T_2,\dots, T_{\lfloor 2\rho \rfloor}$ for $G$.\\
Set $i\leftarrow 0$.\\
Set $X_0\leftarrow E$.\\
Set $F_0 \leftarrow \bigcup_{j=1}^{\lfloor2\rho\rfloor} T_j$.\\
Set $Y_0 \leftarrow X_0 \setminus F_0$.\\
\While{$|Y_i|> 2\rho n$}{
Sample each edge in $Y_i$ with probability $1/2$ to construct $X_{i+1}$.\\
$i \leftarrow i+1$.\\
Set $k_i\leftarrow \rho\cdot 2^{i+1}$.\\
Compute an $k_i$-partial maximum spanning forest packing $T_1, T_2, \dots,T_{k_i}$ for the graph $G_i:=(V,X_i)$. \\
Set $F_i\leftarrow \bigcup_{j=1}^{k_i}T_j $\\
Set $Y_i \leftarrow X_i \setminus F_i$.
}
Set $\Gamma \leftarrow i$. // $\Gamma$ is the number of elapsed iteration in the previous while-loop.\\
Add each edge $e\in Y_\Gamma$ to $G_\epsilon$ with weight $2^{\Gamma-1}w(e)$.\\
Add each edge $e\in F_0$ to $G_\epsilon$ with weight $w(e)$.\\
\For{$j=1,\dots, \Gamma$}{
\ForEach{$e\in F_j$}{Set $p_e \leftarrow \min\left(1,\frac{384}{169}\frac{1}{4^jw(e)}\right)$.\\
Generate $r_e$ from $\textrm{Binom}(2^jw(e),p_e)$.\\
\If{$r_e>0$}{
Add $e$ to $G_\epsilon$ with weight $r_e/p_e$.}}
}
\Return{$G_\epsilon=(V,E_\epsilon)$.}
\end{algorithm}

%\newpage
\subsection{Correctness}
We will prove that $G_\epsilon$ constructed in \textsc{Sparsify}($V,E,w,\epsilon$,c) is a $(1\pm\epsilon)$-cut sparsifier for $G$ with probability at least $1-8/n^c$. Following the proof structure of \cite{FHHP19}, we first define 
\[ S := \left(\bigcup_{i=0}^\Gamma 2^iF_i \right)\cup 2^\Gamma Y_\Gamma, \]
where $\Gamma$ is the maximum number such that $F_i\neq \emptyset$. We define $G_S := (V,S)$. And we prove the following two lemmas, that together yield the desired result. 
\begin{lemma}
\label{thm:GSsparsifiesG}
$G_S$ is a $(1\pm \epsilon/3)$-cut sparsifier for $G$ with probability at least $1-4/n^c$. 
\end{lemma}
\begin{lemma}
\label{thm:GepssparsifiesGS}
$G_\epsilon$ is a $(1\pm \epsilon/3)$-cut sparsifier for $G_S$ with probability at least $1-4/n^c$. 
\end{lemma}

Let us start by proving Lemma~\ref{thm:GSsparsifiesG}. In creating the sets $F_i$, we repeatedly makes use of the MSF indices. The MSF index of an edge immediately ensures a certain connectivity of that edge. The following lemma makes this precise. 
\begin{lemma}\label{lm:heavyinX} 
Let $i\geq 0$ and $e\in Y_i$ be an edge, and set $k_i:=\rho \cdot2^{i+1}$. Then $e$ is $w(e)k_i$-heavy in $G_{i,e}'=(V,X_{i,e}')$, where $X_{i,e}':=\{e'\in X_i : w(e')\geq w(e)\}$. Consequently, $e$ is also $w(e)k_i$-heavy in $G_i=(V,X_i)$.
\end{lemma}
\begin{proof}
Since $e\in Y_i=X_i\setminus F_i$, we know that $e$ was not part of any maximum spanning forest in a $k_i$-partial MSF packing $\mathcal{F}_i$ of $G_i$. Hence, by definition of the maximum spanning forests, each of the forests in $\mathcal{F}_i$ has a path connecting the vertices of $e$, with all edges of weight at least $w(e)$. Thus any cut in $G_i'$ picks up a contribution of at least $w(e)$ for each of the $k_i$ paths. Hence the minimum cut in $G_i'$ separating the vertices of $e$ has value at least $w(e)k_i$, or equivalently $e$ is $w(e)k_i$-heavy in $G_i'$. 
\end{proof}

Next, we show in a general setting that certain ways of sampling preserve cuts. The following lemma is a generalization of Lemma~5.5 in \cite{FHHP19}. 
\begin{lemma}
\label{lm:chernoffforheavyedges}
Let $R\subseteq Q$ be subsets of weighted edges on some set of vertices $V$, satisfying $0<w(e)\leq 1$ for all $e\in Q$. Moreover, assume that each edge in $R$ is $\pi$-heavy in $(V,Q)$. Suppose that each edge $e\in R$ is sampled with probability $p\in(0,1]$, and if selected, given a weight of $w(e)/p$ to form a set of edges $\widehat{R}$. We denote, for every cut $C$:
\[r^{(C)} := \sum_{e\in R\cap C}w(e), \quad\quad q^{(C)} := \sum_{e\in Q\cap C}w(e), \quad \quad \widehat{r}^{(C)} := \sum_{e\in =\widehat{R}\cap C}w(e)/p.\]
Let $\zeta\in \mathbb{N}_{\geq 5}$, and $\delta\in (0,1]$ such that $\delta^2 p\pi\geq \frac{\zeta\ln (n)}{0.38}$, then 
\[ \left| r^{(C)}-\widehat{r}^{(C)} \right| \leq \delta q^{(C)}\]
for all cuts $C$, with probability at least $1-4/n^{\zeta-4}$.
\end{lemma}
\begin{proof}
For each $j\geq 0$, let $\mathcal{C}_j$ be the set consisting of all cuts $C$ with
\[ 2^j\cdot\pi \leq r^{(C)} \leq 2^{j+1}\cdot \pi-1.\]
We will show that for each $j$ the statement of the lemma holds true with probability at least $1-2n^{(4-\zeta)n^j}$. Then the lemma follows from the union bound since
\[\sum_{j=0}^\infty 2n^{(4-\zeta)2^j} \leq 2n^{4-\zeta}\sum_{j=0}^\infty 2^{-(2^j-1)} \leq 2n^{4-\zeta}\sum_{k=0}^\infty 2^{-k} \leq 4n^{4-\zeta},\]
where we use that $n^{4-\zeta}\leq1/2$.

Let $C\in \mathcal{C}_j$. For every $e\in R$, define the random variables $Y_e$ that takes value $w(e)$ with probability $p$ and 0 otherwise. We have $Y_e\in[0,1]$, $\E[Y_e]=pw(e)$, and $\sum_{e\in R} Y_e = pr^{(C)}$. Now we apply Theorem~\ref{thm:Chernoff} with $\epsilon = \delta q^{(C)}/r^{(C)}$ and $\mu = pr^{(C)}$ to obtain
\begin{align*}
    \P\left[ \left| r^{(C)}-\widehat{r}^{(C)} \right| > \delta q^{(C)}\right] &= \P\left[\left|\sum_{e\in R} Y_e-\mu\right|>\delta\frac{q^{(C)}}{r^{(C)}} \cdot p r^{(C)}\right]\\
    &\leq 2\exp\left(-0.38\delta^2\left(\frac{q^{(C)}}{r^{(C)}}\right)^2pr^{(C)}\right)\\
    &\leq 2\exp\left(-0.38 \delta^2 pq^{(C)}\right),
\end{align*}
where the last inequality holds as $r^{(C)}\leq q^{(C)}$ since $R\subseteq Q$. Now observe that $q^{(C)}\geq r^{(C)}\geq \pi\cdot 2^j$, hence 
\begin{align*} 
 \P\left[ \left| r^{(C)}-\widehat{r}^{(C)} \right| > \delta q^{(C)}\right] &\leq 2\exp\left(-0.38 \delta^2 p\pi2^j\right)\\
 &\leq 2\exp\left(-\zeta\ln (n) 2^j\right)\\
&= n^{-\zeta2^j}. 
\end{align*}
As every edge in $R\cap C$ is $\pi$-heavy in $(V,Q)$, we can apply Theorem~\ref{thm:cutprojection} to see that the number of distinct sets $R\cap C$ is at most:
\[ n^{2\frac{2^{j+1}\pi}{\pi}} = n^{4\cdot2^j}.\]
Thus the union bound gives us that the statement of the lemma holds true for all cuts $C\in\mathcal{C}_j$ with probability at least $1-2n^{(4-\zeta)2^j}$.
\end{proof}

We want to apply this lemma to our sampling procedure. We do this by considering different weight classes separately. We define $X_{i,k}:= \{ e\in X_{i} : 2^k \leq w(e) \leq 2^{k+1}-1\}$, and $x_{i,k}^{(C)} = \sum_{e\in X_{i,k}\cap C}w(e)$. We define $Y_{i,k}$ and $y^{(C)}_{i,k}$ analogously. Some rescaling is necessary to ensure that all weights lie in $(0,1]$, as Lemma~\ref{lm:chernoffforheavyedges} requires. For $A\subseteq E$ and $\beta>0$, we write $\beta A$ to indicate we multiply the weight of the edges by a factor of $\beta$.  
\begin{lemma}
\label{lm:onestepcutpreservationweightclasses}
With probability at least $1-4/n^{4+c}$, for every cut $C$ in $G_i$, 
\begin{align*}
    \left|2^{-k}x_{i+1,k}^{(C)}-2^{-k-1}y^{(C)}_{i,k}\right|\leq \frac{\epsilon/13}{2^{i/2+1}}\sum_{k'= k}^\infty 2^{-k'-1}x_{i,k'}^{(C)}.
\end{align*} 
\end{lemma}
\begin{proof}
Any $e\in Y_{i,k}$ is $\rho\cdot 2^{i+1}w(e)\geq \rho\cdot 2^{i+k+1}$-heavy in $\bigcup_{k'= k}^\infty X_{i,k'}$. A closer look shows us that we also have that any $e\in 2^{-k-1}Y_{i,k}$ is $\rho \cdot2^{i}$-heavy in $\bigcup_{k'= k}^\infty 2^{-k'-1}X_{i,k'}$. We set $R= 2^{-k-1}Y_{i,k}$, $Q=\bigcup_{k'= k}^\infty 2^{-k'-1}X_{i,k'}$, $\pi = \rho\cdot 2^i$, $p=1/2$, and $\delta= \frac{\epsilon/13}{2^{i/2+1}}$, and we check that
\begin{align*}
    \delta^2 p \pi = \frac{\epsilon^2/13^2}{2^{i+3}}\rho 2^i
    =\frac{\epsilon^2}{2^3\cdot13^2}\frac{(7+c)1352 \ln (n)}{0.38 \epsilon^2}
    = \frac{(7+c) \ln (n)}{0.38}.
\end{align*}
So we can apply Lemma~\ref{lm:chernoffforheavyedges} with these settings to obtain:
\begin{align*}
    \left|2^{-k}x_{i+1,k}^{(C)}-2^{-k-1}y^{(C)}_{i,k}\right|\leq \frac{\epsilon/13}{2^{i/2+1}}\sum_{k'= k}^\infty 2^{-k'-1}x_{i,k'}^{(C)},
\end{align*}
which holds for all cuts $C$ with probability $1-4/n^{3+c}$.
\end{proof}

Now we look at the general case, for which we sum all weight classes. Hereto, we define $x_i^{(C)} = \sum_{e\in X_i\cap C}w(e)$, $x_{i+1}^{(C)} = \sum_{e\in X_{i+1}\cap C}w(e)$, and $y_i^{(C)} = \sum_{e\in Y_i\cap C}w(e)$.
\begin{corollary}
\label{lm:onestepcutpreservation}
 With probability at least $1-4/n^{1+c}$, for every cut $C$ in $G_i$,
\[ \left|2x^{(C)}_{i+1}-y_i^{(C)}\right|\leq \frac{\epsilon/13}{2^{i/2}}\cdot x_i^{(C)}.\]
\end{corollary}
\begin{proof}
We rescale and sum over $k$ for each of the weight classes in Lemma~\ref{lm:onestepcutpreservationweightclasses} to get 
\begin{align*}
    \left|2x^{(C)}_{i+1}-y_i^{(C)}\right| &= \left|\sum_{k=0}^\infty 2^{k+1}\left( 2^{-k}x^{(C)}_{i+1,k}-2^{-k-1}y_{i,k}^{(C)}\right)\right|\\
    &\leq \sum_{k=0}^\infty 2^{k+1} \left| 2^{-k}x^{(C)}_{i+1,k}-2^{-k-1}y_{i,k}^{(C)}\right|\\
    &\leq \sum_{k=0}^\infty 2^{k+1} \left(\frac{\epsilon/13}{2^{i/2+1}}\sum_{k'= k}^\infty 2^{-k'-1}x_{i,k'}^{(C)}\right)
\end{align*}
Next, we want to interchange the sum over $k$ with the sum over $k'$, a visual argument for the adjustment of the bounds can be found in Figure~\ref{fig:sums0}.

\begin{figure}[!t]
\centering
\begin{tikzpicture}[scale=0.6]
  \draw[->] (-1, 0) -- (9, 0) node[right] {$k$};
  \fill[blue!10, domain=0:5]
(0,0) -- (8,8) -- (0,8);
  \draw[->] (0, -1) -- (0, 9) node[above] {$k'$};
  \draw[domain=0:8, smooth, variable=\x, blue]  plot ({\x}, {\x});
\end{tikzpicture}
\caption{A visualization of the area covered by $\sum_{k=0}^\infty \sum_{k'= k}^\infty  1= \sum_{k'= 0}^\infty \sum_{k=0}^{k'} 1$.}
\label{fig:sums0}
\end{figure}
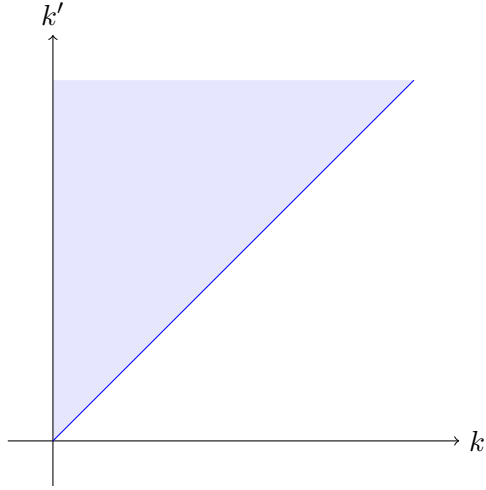

\begin{align*}
    \sum_{k=0}^\infty 2^{k+1} \left(\frac{\epsilon/13}{2^{i/2+1}}\sum_{k'= k}^\infty 2^{-k'-1}x_{i,k'}^{(C)}\right) &= \frac{\epsilon/13}{2^{i/2+1}}\sum_{k'= 0}^\infty 2^{-k'-1}x_{i,k'}^{(C)}\sum_{k=0}^{k'} 2^{k+1}\\
    &\leq  \frac{\epsilon/13}{2^{i/2+1}}\sum_{k'= 0}^\infty 2^{-k'-1}x_{i,k'}^{(C)}2^{k'+2}\\
    &\leq  \frac{\epsilon/13}{2^{i/2}} x_i^{(C)},
\end{align*}
which holds simultaneously for all cuts $C$ with probability at least $1-4/n^{1+c}$. The reason is that at most $m\leq n^2$ of the $X_{i,k}\cap C$ are non-empty, hence a union bound gives the desired bound on the probability. 
\end{proof}

We will repeatedly apply this lemma. To show that the accumulated error does not grow beyond $\epsilon/3$, we use the following fact. For a proof we refer to \cite{FHHP19}. 

\begin{lemma}%[Fact 5.7 in \cite{FHHP19}]
\label{lm:erroraccumulation}
Let $x\in(0,1]$ be a parameter. Then for any $k\geq 0$,
\begin{align*}
    \prod_{i=0}^k\left(1+\frac{x/13}{2^{i/2}}\right) &\leq 1+x/3,\\
    \prod_{i=0}^k\left(1-\frac{x/13}{2^{i/2}}\right) &\geq 1-x/3.
\end{align*}
\end{lemma}

As a final step towards proving Lemma~\ref{thm:GSsparsifiesG}, we prove a lemma that focusses on the sparsification occurring in the last $\Gamma-j+1$ iterative steps of our algorithm.
\begin{lemma}\label{lm:SjsparsifiesGj}
Let 
\[ S_j = \left(\bigcup_{i=j}^\Gamma 2^{i-j}F_i\right) \cup 2^{\Gamma-j}Y_\Gamma\]
for any $j\geq 0$. Then, $S_j$  is a $\left(1\pm (\epsilon/3)2^{-j/2}\right)$-cut sparsifier for $G_j=(V,X_j)$, with probability at least $1-4/n^c$.
\end{lemma}
Note that setting $j=0$ gives us Lemma~\ref{thm:GSsparsifiesG}. Although this lemma is a generalization of the corresponding case for unweighted graphs in \cite{FHHP19}, the proof for the weighted case will be exactly the same: all the work that needed to be done is contained in the previous lemmas. We include the proof here for completeness. 
\begin{proof}[Proof of Lemma~\ref{lm:SjsparsifiesGj}]
Let $C$ be a cut. We define $s_j^{(C)}:= \sum_{e\in S_j \cap C} w_{S_j}(e)$, and $f_i^{(C)}$, $x_i^{(C)}$, and $y_i^{(C)}$ analogously. We will show that the weight of $C$ in $S_j$ is at most $(1+(\epsilon/3)2^{-j/2})$ times the weight of $C$ in $G_j$.

\begin{align*}
s_j^{(C)} &= \sum_{i=j}^\Gamma 2^{i-j}f_i^{(C)} + 2^{\Gamma-j} y_\Gamma^{(C)} &\text{by definition of }S_j\\
&= \sum_{i=j}^{\Gamma-1} 2^{i-j}f_i^{(C)} + 2^{\Gamma-j} x_\Gamma^{(C)} &\text{by definition of }Y_\Gamma\\
&\leq \sum_{i=j}^{\Gamma-1} 2^{i-j}f_i^{(C)} + 2^{\Gamma-j-1} \left(y_{\Gamma-1}^{(C)}+\frac{\epsilon/13}{2^{(\Gamma-1)/2}}x_{\Gamma-1}^{(C)}\right) &\text{by Lemma~\ref{lm:onestepcutpreservation}}\\
&= \sum_{i=j}^{\Gamma-2} 2^{i-j}f_i^{(C)} + 2^{\Gamma-j-1} x_{\Gamma-1}^{(C)}\left(1+\frac{\epsilon/13}{2^{(\Gamma-1)/2}}\right) &\text{by definition of }Y_{\Gamma-1}%\\
\end{align*}
We repeat the last step $\Gamma-j-1$ times to conclude
\begin{align*}
s_j^{(C)} &\leq x_j^{(C)}\prod_{i=j}^{\Gamma-1}\left(1+\frac{\epsilon2^{-j/2}/13}{2^{i/2}}\right)\\
&\leq  x_j^{(C)}(1+(\epsilon/3)2^{-j/2}) &\text{by Lemma \ref{lm:erroraccumulation}}
\end{align*}
The proof of $s_j^{(C)} \geq x_j^{(C)}(1-(\epsilon/3)2^{-j/2})$ is analogous. As we have that $\Gamma\leq n$, we can use a union bound to conclude that Lemma~\ref{lm:onestepcutpreservation} holds for all simultaneously with probability at least $1-4/n^c$, which concludes the proof. 
\end{proof} 

To prove Lemma~\ref{thm:GepssparsifiesGS}, we will invoke the framework from \cite{FHHP19}, as given in Section~\ref{subsc:framework}. More specifically, we will apply Theorem~\ref{thm:framework}. We set the parameter $\gamma:=64/3$, and for each $e\in F_i$ we set $\lambda_e := \rho\cdot 4^iw(e)$. This is in line with our choice for $p_e$:
\[\min\left(1,\frac{16(c+7)\gamma \ln (n)}{0.38 \lambda_e \epsilon^2}\right) = \min\left(1,\frac{16(c+7)\gamma \ln (n)}{0.38 \rho\cdot 4^iw(e)_e \epsilon^2}\right)=  \min\left(1,\frac{384}{169}\frac{1}{4^iw(e)}\right) = p_e .\]
We have to provide a set of subgraphs $\mathcal G$ and a set of parameters $\Pi$ such that $\Pi$-connectivity and $\gamma$-overlap are satisfied. 

To explore the connectivity of edges in $R_i:= \{e\in E : 2^i\leq \lambda_e\leq 2^{i+1}-1\}$ we partition these sets as follows:
\begin{align*}
    R_{j,k} := \{ e\in F_j : 2^k \leq \rho w(e)\leq 2^{k+1}-1\}. 
\end{align*}
We will view these edges in the subgraph:
\begin{align*}
    E_{j,k} := \bigcup_{j'=j-1}^\Gamma \bigcup_{k'=k}^\infty \rho \cdot 4^{\Gamma-j'+1}2^{\Lambda-k'+j'} R_{j',k'}.
\end{align*}

\begin{lemma}\label{lm:RjkheavyinEjk}
Each edge $e\in R_{j,k}$ is $\pi:=\rho \cdot 4^\Gamma2^\Lambda$-heavy in $(V,E_{j,k})$.
\end{lemma}
\begin{proof}
Fix $e\in R_{j,k}$. This edge is $\rho\cdot 2^{j}w(e)\geq \rho \cdot 2^{j+k}$-heavy in $\{e\in X_{j-1}: w(e') \geq w(e)\}$, see Lemma~\ref{lm:heavyinX}. Hence $e$ is $\rho \cdot 2^{j+k}$-heavy in $\{e'\in X_{j-1} : \rho w(e') \geq 2^k\}$. We can rescale this: $e$ is {$\left(\rho \cdot 2^{j+\Lambda}\right)$-heavy} in $2^{\Lambda-k}\cdot\{e'\in X_{j-1} : \rho w(e') \geq 2^k\}=\bigcup_{k'=k}^\infty2^{\Lambda-k}\cdot\{e'\in X_{j-1} : 2^{k'}\leq \rho w(e') \leq 2^{k'+1}-1\} $. We rescale again to see $e$ is $\rho \cdot 2^{2j+\Lambda-1}$-heavy in $\bigcup_{k'=k}^\infty2^{\Lambda-k+j-1}\cdot\{e'\in X_{j-1} : 2^{k'}\leq \rho w(e') \leq 2^{k'+1}-1\} $. Next, we want to replace $X_{j-1}$ with $S_{j-1}$. Hereto, we apply Lemma~\ref{lm:onestepcutpreservationweightclasses} with $\epsilon = 13\cdot 2^{i/2+1}$, which shows that for each of the weight classes the cuts are preserved up to a factor $2$. Hence we obtain $e$ is $\rho \cdot 2^{2j+\Lambda-2}$-heavy in $\widetilde{E}_{j,k}:=\bigcup_{j'= j-1}^\Gamma \bigcup_{k'=k}^\infty 2^{\Lambda-k+j'}\cdot\{e'\in F_{j'} : 2^{k'}\leq \rho w(e') \leq 2^{k'+1}-1\} $. 

Now let $e'\in R_{j,k}$ be any edge, and let $C$ be a cut such that $e'\in C$. We need to show that the weight of this cut in $E_{j,k}$ is at least $\rho\cdot4^\Gamma2^\Lambda$. Let $e:= \argmin_{e \in C} \{j_e : e\in R_{j_e,k_e} \text{ for some }k_e\geq k\}$ (in case $e$ is not unique, pick any). By the above statement we have that $e$ is $\rho \cdot 2^{2j_e+\Lambda-2}$-heavy in $\widetilde{E}_{j_e,k_e}\subseteq \widetilde{E}_{j_e,k}$. Thus $e$ is $\rho \cdot 4^\Gamma2^\Lambda$-heavy in $4^{\Gamma-j_e+1}\widetilde{E}_{j_e,k}$. This is a subgraph of $E_{j_e,k}$, which in turn is a subgraph of $E_{j,k}$. Hence $e$ is $\rho \cdot 4^\Gamma2^\Lambda$-heavy in $E_{j,k}$, and thus $C$ has weight at least $\rho \cdot 4^\Gamma2^\Lambda$.
\end{proof}

Now we take all weight classes together to find the set of subgraphs $\mathcal G$ for which $\Pi$-connectivity is satisfied. 

\begin{corollary}\label{cor:piconnectivity}
Each edge in $e\in R_i$ is $\rho \cdot 4^\Gamma2^\Lambda$-heavy in $G_i=(V,E_i)$, with $E_i := \bigcup_{j=1}^{\min(\lfloor i/2\rfloor,\Gamma)} E_{j,i-2j}$.
\end{corollary}
\begin{proof}
Note that $e\in R_i$ satisfies $2^i\leq \rho \cdot 2^{2j} w(e)\leq 2^{i+1}-1$ if $e\in F_j$. Hence $e\in R_{j,k}$ with $2j+k=i$. We are only considering edges in $F_j$ with $1\leq j\leq \Gamma$, thus we have $R_i = \bigcup_{j=1}^{\min(\lfloor i/2\rfloor,\Gamma)} R_{j,i-2j}$, hence the claim follows directly from Lemma~\ref{lm:RjkheavyinEjk}.
\end{proof}

It remains to show that $\gamma$-overlap is satisfied. 
\begin{lemma}\label{lm:gammaoverlap}
For any cut $C$, 
    \[ \sum_{i=0}^\Lambda \frac{e_i^{(C)}2^{i-1}}{\rho\cdot 4^\Gamma 2^\lambda}\leq 64/3 \cdot e^{(C)},\]
    where $e^{(C)}=\sum_{e\in C}w_{G_S}(e)$ and $e_i^{(C)}=\sum_{e\in C\cap E_i}w_{G_i}(e)$.
\end{lemma}
\begin{proof}
We add $F_0$ and $Y_\Gamma$ to $G_\epsilon$, so we do not need to be concerned about the intersection of the cut $C$ with these sets. This means we only intersect a cut $C$ with $F_j$ where $1\leq j \leq \Gamma$. Hence we start our sum with $i=2$. We consider the sum we need to bound:
\begin{align*}
	\sum_{i=2}^\Lambda \frac{e_i^{(C)}2^{i-1}}{\rho\cdot 4^\Gamma 2^\lambda} &= \sum_{i=2}^\Lambda \frac{\left(\sum_{e\in C\cap E_i}w_{G_i}(e)\right)2^{i-1}}{\rho\cdot 4^\Gamma 2^\lambda} \\
&= \sum_{i=2}^\Lambda \sum_{j=1}^{\min(\lfloor i/2\rfloor, \Gamma)} \frac{\left(\sum_{e\in C\cap E_{j,i-2j}}w_{G_i}(e)\right)2^{i-1}}{\rho\cdot 4^\Gamma 2^\lambda} \\
&= \sum_{i=2}^\Lambda \sum_{j=1}^{\min(\lfloor i/2\rfloor, \Gamma)} \sum_{j'=j-1}^\Gamma \sum_{k'=i-2j}^\infty \frac{ \rho \cdot4^{\Gamma-j'+1}2^{\Lambda-k'+j'}\left(\sum_{e\in C\cap E_{j',k'}}w_{G}(e)\right)2^{i-1}}{\rho\cdot 4^\Gamma 2^\lambda} \\
&= \sum_{i=2}^\Lambda \sum_{j=1}^{\min(\lfloor i/2\rfloor, \Gamma)} \sum_{j'=j-1}^\Gamma \sum_{k'=i-2j}^\infty  2^{-k'-j'+i+1}\left(\sum_{e\in C\cap E_{j',k'}}w_{G}(e)\right).
\end{align*}
Next, we want to interchange the sum over $i$ and the sum over $j$ and change the bounds accordingly. See Figure \ref{fig:sums12a} for a visual argument. 

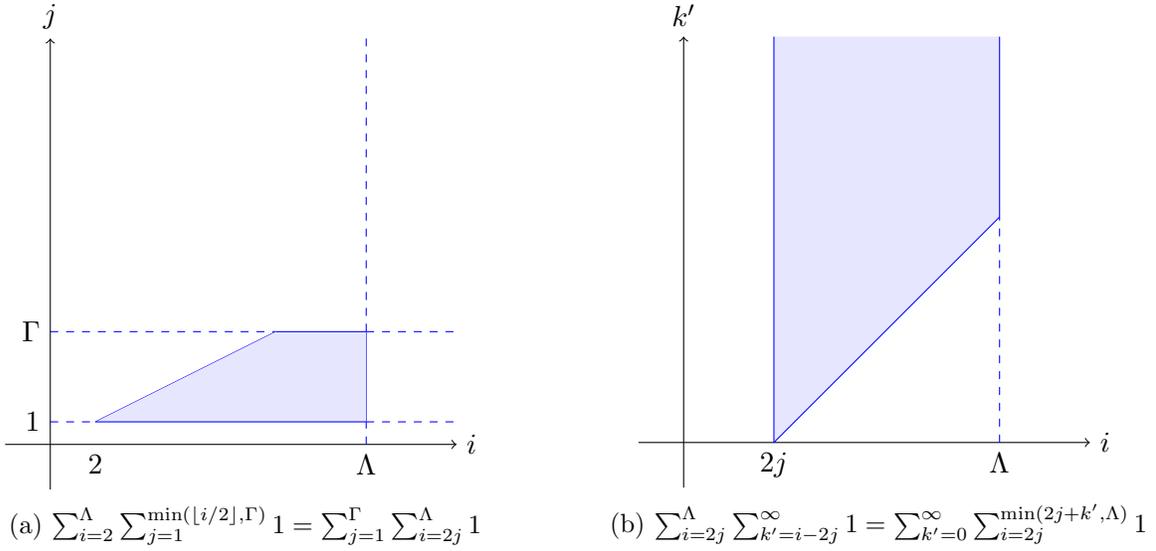
\begin{figure}[!h]
\centering
\begin{subfigure}{.49\textwidth}
  \centering
\begin{tikzpicture}[scale=0.6]
  \draw[->] (-1, 0) -- (9, 0) node[right] {$i$};
  \draw[->] (0, -1) -- (0, 9) node[above] {$j$};
  \draw[scale=1.0, domain=1:5, smooth, variable=\x, blue] plot ({\x}, {\x/2});
  \draw[domain=0:5,dashed, variable=\x,blue] plot ({\x},2.5);
  \draw[domain=7:9,dashed, variable=\x,blue] plot ({\x},2.5);
  \draw[domain=0:1,dashed, variable=\x,blue] plot ({\x},0.5);
  \draw[domain=7:9,dashed, variable=\x,blue] plot ({\x},0.5);
  \draw[domain=0.5:2.5, smooth, variable=\y, blue]  plot (7, {\y});
  \draw[domain=0:0.5, dashed, variable=\y, blue]  plot (7, {\y});
  \draw[domain=2.5:9, dashed, variable=\y, blue]  plot (7, {\y});
  \draw (0,2.5) node[left] {$\Gamma$};
  \draw (7,0) node[below] {$\Lambda$};
  \draw (1,0) node[below] {$2$};
  \draw (0,0.5) node[left] {$1$};
  \fill[blue!10, domain=0:5, variable =\x]
(1,0.5) -- plot  ({\x}, {\x/2}) -- (7,2.5) -- (7,0.5) -- cycle;
  \draw[domain=1:7,smooth, variable=\x,blue] plot ({\x},0.5);
  \draw[domain=5:7,smooth, variable=\x,blue] plot ({\x},2.5);
\end{tikzpicture}
\caption{$\sum_{i=2}^\Lambda \sum_{j=1}^{\min(\lfloor i/2\rfloor, \Gamma)}1=\sum_{j=1}^\Gamma \sum_{i=2j}^{\Lambda}1$}
\label{fig:sums12a}
\end{subfigure}\hfill
\begin{subfigure}{.49\textwidth}
  \centering
\begin{tikzpicture}[scale=0.6]
  \draw[->] (-1, 0) -- (9, 0) node[right] {$i$};
  \draw[->] (0, -1) -- (0, 9) node[above] {$k'$};
  \draw (7,0) node[below] {$\Lambda$};
  \draw (2,0) node[below] {$2j$};
\fill[blue!10, domain=0:5]
(2,9) --  (2,0) -- (7,5) -- (7,9);
  \draw[scale=1.0, domain=2:7, smooth, variable=\x, blue] plot ({\x}, {\x-2});
  \draw[domain=0:5, dashed, variable=\y, blue]  plot (7, {\y});
  \draw[domain=5:9, smooth, variable=\y, blue]  plot (7, {\y});
  \draw[domain=0:9, smooth, variable=\y, blue]  plot (2, {\y});
\end{tikzpicture}
\caption{$\sum_{i=2j}^{\Lambda}\sum_{k'=i-2j}^\infty 1= \sum_{k'=0}^\infty \sum_{i=2j}^{\min(2j+k',\Lambda)}1 $}
\label{fig:sums12b}
\end{subfigure}
\caption{Two visualizations of the area covered by a double sum.}
\label{fig:sums12}
\end{figure}

\begin{align*}
&\sum_{i=2}^\Lambda \sum_{j=1}^{\min(\lfloor i/2\rfloor, \Gamma)} \sum_{j'=j-1}^\Gamma \sum_{k'=i-2j}^\infty  2^{-k'-j'+i+1}\left(\sum_{e\in C\cap E_{j',k'}}w_{G}(e)\right) \\
&=\sum_{j=1}^\Gamma \sum_{i=2j}^{\Lambda} \sum_{j'=j-1}^\Gamma \sum_{k'=i-2j}^\infty  2^{-k'-j'+i+1}\left(\sum_{e\in C\cap E_{j',k'}}w_{G}(e)\right) 
\end{align*}

Interchanging the sum over $i$ and $j'$ does not change the bounds, as they are independent of each other. When interchanging the sum over $i$ and the sum over $k'$ we have to be more careful, see Figure \ref{fig:sums12b} for a visual argument. 

\begin{align*}
&\sum_{j=1}^\Gamma \sum_{i=2j}^{\Lambda} \sum_{j'=j-1}^\Gamma \sum_{k'=i-2j}^\infty  2^{-k'-j'+i+1}\left(\sum_{e\in C\cap E_{j',k'}}w_{G}(e)\right) \\
&=\sum_{j=1}^\Gamma  \sum_{j'=j-1}^\Gamma \sum_{k'=0}^\infty \sum_{i=2j}^{\min(2j+k',\Lambda)}  2^{-k'-j'+i+1}\left(\sum_{e\in C\cap E_{j',k'}}w_{G}(e)\right)\\
 &\leq \sum_{j=1}^\Gamma  \sum_{j'=j-1}^\Gamma \sum_{k'=0}^\infty  2^{-k'-j'+2j+k'+2}\left(\sum_{e\in C\cap E_{j',k'}}w_{G}(e)\right)\\
&= \sum_{j=1}^\Gamma  \sum_{j'=j-1}^\Gamma  2^{2j-j'+2}\left(\sum_{k'=0}^\infty \sum_{e\in C\cap E_{j',k'}}w_{G}(e)\right)\\
&= \sum_{j=1}^\Gamma  \sum_{j'=j-1}^\Gamma  2^{2j-j'+2}\left(\sum_{e\in C\cap F_{j'}}w_{G}(e)\right)
\end{align*}
Next, we want to interchange the sum over $j$ with the sum over $j'$, a visual argument can be found in Figure \ref{fig:sums3}.

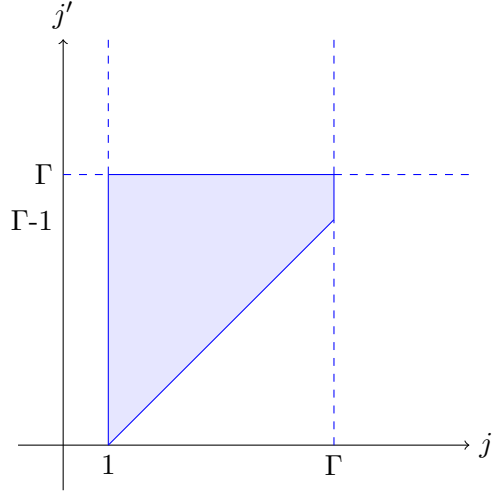
\begin{figure}[!h]
\centering
\begin{tikzpicture}[scale=0.6]
  \draw[->] (-1, 0) -- (9, 0) node[right] {$j$};
  \draw[->] (0, -1) -- (0, 9) node[above] {$j'$};
  \draw (6,0) node[below] {$\Gamma$};
  \draw (1,0) node[below] {$1$};
  \draw (0,6) node[left] {$\Gamma$};
  \draw (0,5) node[left] {$\Gamma$-1};
\fill[blue!10, domain=0:5]
(1,0) -- (6,5) -- (6,6) -- (1,6);
  \draw[domain=0:5, dashed, variable=\y, blue]  plot (6, {\y});
  \draw[domain=5:6, smooth, variable=\y, blue]  plot (6, {\y});
  \draw[domain=6:9, dashed, variable=\y, blue]  plot (6, {\y});
  \draw[domain=0:1, dashed, variable=\x, blue]  plot ({\x}, 6);
  \draw[domain=1:6, smooth, variable=\x, blue]  plot ({\x}, 6);
  \draw[domain=6:9, dashed, variable=\x, blue]  plot ({\x}, 6);
  \draw[domain=1:6, smooth, variable=\x, blue]  plot ({\x}, {\x-1});
  \draw[domain=0:6, smooth, variable=\y, blue]  plot (1, {\y});
  \draw[domain=6:9, dashed, variable=\y, blue]  plot (1, {\y});
\end{tikzpicture}
\caption{A visualization of the area covered by $\sum_{j=1}^\Gamma  \sum_{j'=j-1}^\Gamma 1=  \sum_{j'=0}^\Gamma  \sum_{j=1}^{j'+1}1$.}
\label{fig:sums3}
\end{figure}

\begin{align*}
\sum_{j=1}^\Gamma  \sum_{j'=j-1}^\Gamma  2^{2j-j'+2}\left(\sum_{e\in C\cap F_{j'}}w_{G}(e)\right) &= \sum_{j'=0}^\Gamma  2^{-j'+2}\sum_{j=1}^{j'+1}  4^{j}\left(\sum_{e\in C\cap F_{j'}}w_{G}(e)\right)\\
&\leq \sum_{j'=0}^\Gamma  2^{-j'+2}\frac{4^{j'+2}}{3}\left(\sum_{e\in C\cap F_{j'}}w_{G}(e)\right)\\
&= \frac{64}{3} \sum_{j'=0}^\Gamma 2^{j'}\left(\sum_{e\in C\cap F_{j'}}w_{G}(e)\right)\\
&= \frac{64}{3} \sum_{e\in C} w_{G_S}(e)\\
&= \frac{64}{3} e^{(C)}.
\qedhere
\end{align*}
\end{proof}

Together Corollary~\ref{cor:piconnectivity} and Lemma~\ref{lm:gammaoverlap} show that the conditions of Theorem~\ref{thm:framework} are met with the given parameters. This proves Lemma~\ref{thm:GepssparsifiesGS}, and then Theorem~\ref{thm:ourcontribution} follows.

\subsection{Size of the Sparsifier}
\label{subsc:size}
The sparsifier $G_\epsilon$ consists of $F_0$, $Y_\Gamma$, and $F '$, where $F'= \cup_{i=1} ^\Gamma F '_i$, with $ F'_i$ the sampled edges of $F_i$. First of all, note that $|F_0|= O(cn\ln (n)/\epsilon^2)$ and $|Y_\Gamma|=O(cn\ln (n)/\epsilon^2) $. Now take $e\in F_i$. This edge results to an edge in $G_\epsilon$ if the sample from the binomial distribution with parameters $n_e=2^iw(e)$ and $p_e= \min\left(1,\frac{384}{169}\frac{1}{4^iw(e)}\right)$ is positive. The probability that this happens is
\begin{align*}
\P[ \text{Binom}(n_e,p_e) > 0] &= \sum_{k=1}^{n_e}\P[ \text{Binom}(n_e,p_e) = k]\\
&\leq \sum_{k=1}^{n_e}k\P[ \text{Binom}(n_e,p_e) = k]\\
&= \sum_{k=0}^{n_e}k\P[ \text{Binom}(n_e,p_e) = k]\\
&= \E[ \text{Binom}(n_e,p_e)]\\
&= n_ep_e\\
&\leq \frac{384}{169}2^{-i}.
\end{align*}
Note that this probability is equal for all $e\in F_i$. Since $F_i$ is the union of $k_i=\rho\cdot 2^{i+1}$ spanning forests, we know that $|F_i|\leq \rho 2^{i+1}n$. Hence the expected size of $F'_i$, the sampled edges in $F_i$, equals 
\begin{align*}
	\E[|F'_i|] &= \sum_{e\in F_i} \P[ \text{Binom}(n_e,p_e) > 0] \\
&\leq \sum_{e\in F_i}  \frac{384}{169}2^{-i}\\
&= |F_i|  \frac{384}{169}2^{-i}\\
&\leq  \rho 2^{i+1}n \frac{384}{169}2^{-i}\\
&=  \rho  \frac{768}{169} n.
\end{align*}
We have that the total number of sampled edges equals 
\[ \E[|F'|] = \sum_{i=1}^\Gamma \E[|F'_i|] \leq \Gamma\rho \frac{768}{169}n,\]
so it remains to bound $\Gamma$, i.e., the number of $F_i$'s. Hereto, note that the while loop of lines 10--17 ends if $|Y_i| \leq 2\rho n$. We bound the number of edges in $Y_i$ by bounding the number of edges of $X_i$, of which $Y_i$ is a subset. Each edge in $Y_{i-1}\subseteq X_{i-1}$ is sampled with probability $1/2$ to form $X_i$. So $\E[|X_i|]\leq |X_{i-1}|/2$. Now by a Chernoff bound (see Theorem~\ref{thm:Chernoffupperbound}) we obtain:
\[ \P\left[ |X_i|> \frac{2}{3}|X_{i-1}|\right] \leq \exp\left(-\frac{0.38}{36}|X_{i-1}|\right) > \exp\left(-\frac{cn\ln (n)}{36}\right) = n^{-cn/36},\] 
since $|X_{i-1}|\geq|Y_{i-1}|\geq 2\rho n=2\cdot \frac{(7+c)1352 \ln (n)}{0.38 \epsilon^2}n\geq \frac{cn\ln (n)}{0.38}$. We have at most $n^2$ sets $X_i$, so we can conclude that with high probability $|X_i|\leq \frac{2}{3}|X_{i-1}|$ in each step, and by induction $|Y_i|< |X_i|\leq \left(\frac{2}{3}\right)^i m$. We see that
\begin{align*}
 m\left(\frac{2}{3}\right)^\Gamma  &\leq 2\rho n= \frac{21632}{0.38\epsilon^2}cn \ln (n),
\end{align*}
which is equivalent to
\begin{align*}
 \left(\frac{2}{3}\right)^\Gamma &\leq \frac{\frac{21632}{0.38\epsilon^2}cn \ln (n)}{m},
\end{align*}
and that is equivalent to
\begin{align*}
\Gamma &\geq \log\left( \frac{m}{\frac{21632}{0.38\epsilon^2}cn \ln (n)}\right) /\log(3/2).
\end{align*}
So, we can conclude $\Gamma = O{\left(\log \left(\frac{m}{cn\log (n)/\epsilon^2}\right)\right)}$. This gives that the total number of sampled edges is, in expectation, 
\[ \E[|F'|]\leq \Gamma\rho \frac{768}{169}n= O(cn\log (n)\log\left(m/(cn\log (n)/\epsilon^2)\right)/\epsilon^2).\]
This compression process can also be seen as the sum of $m$ independent random variables that take values in $\{1,0\}$.\footnote{To be precise, we set the probability of an edge $e\notin \bigcup_i F_i$ to exist to $0$.} We have just calculated that the expected value $\mu$ is at most $Bcn\ln (n)\log\left(m/(cn\log (n)/\epsilon^2)\right)/\epsilon^2$, for some $B>0$. Using this, we apply a Chernoff bound (Theorem~\ref{thm:Chernoffupperbound}) to get an upper limit for the number of sampled edges:
\begin{align*}
\P\left[ |F'| > 2B cn\ln n\log\left(m/(cn\log (n)/\epsilon^2)\right)/\epsilon^2\right] &\leq \exp\left(-0.38 B cn\ln (n)\log\left(m/(cn\log (n)/\epsilon^2)\right)/\epsilon^2\right) \\
&= n^{-0.38 cnB\log\left(m/(cn\log (n)/\epsilon^2)\right)/\epsilon^2}. 
\end{align*}
We conclude that, with high probability, the number of sampled edges is 
\begin{align*}
O(2Bcn\ln (n)\log\left(m/(cn\log (n)/\epsilon^2)\right)/\epsilon^2)=O(cn\log (n)\log\left(m/(cn\log (n)/\epsilon^2)\right)/\epsilon^2).
\end{align*}
And finally, we conclude that with high probability the number of edges of $G_\epsilon$ is bounded by $|E(G_\epsilon)|=|F_0|+|Y_\Gamma|+|F'|=O(cn\log (n)\log\left(m/(cn\log (n)/\epsilon^2)\right)/\epsilon^2)$.

\subsection{Time Complexity}
\label{subsc:time}
First off, if $m\leq 4\rho n\log\left(m/(n\log (n)/\epsilon^2)\right) = O(cn\log (n)/\epsilon^2\log\left(m/(n\log (n)/\epsilon^2)\right))$, the algorithm does nothing and returns the original graph. So for this analysis we can assume $m>4\rho n\log\left(m/(n\log (n)/\epsilon^2)\right)$. We analyze the time complexity of the algorithm in two phases. The first phase consists of computing the probabilities $p_e$ for all $e\in E$. The second one is compressing edges, given these probabilities.

The first phase contains $i$ iterations of the while loop (lines 10--17). In each iteration we sample edges from $Y_i\subseteq X_{i}$ with probability $1/2$ to form $X_{i+1}$. This takes time at most $O(|X_{i}|)$. Next, we compute a maximum spanning forest packing of the graph $G_{i+1}=(V,X_{i+1})$. We know that we can compute a $M$-partial maximum spanning forest packing of a polynomially-weighted graph with $n$ vertices and $m_0$ edges in $O(m_0\cdot\min(\alpha(n)\log (M),\log (n)) )$ time (see Theorem~\ref{thm:MSFP} and Theorem~\ref{thm:MSFdense}). So this iteration takes at most $O(|X_{i+1}|\cdot(\min(\alpha(n)\log (k_{i+1}),\log (n) )))$ time. As noted earlier, we have with high probability that $|X_{i}| \leq \left(\frac{2}{3}\right)^{i} m$. If $m\alpha(n)\log(m/n)\leq m \log (n)$, we conclude w.h.p.\ that the first phase takes total time at most
\begin{align*}
\sum_{i=0}^\Gamma  O(|X_{i}|)+O(|X_{i+1}|\alpha(n)\log (k_{i+1}) ) &= \sum_{i=0}^\Gamma  \left(\frac{2}{3}\right)^{i}O(m)+\left(\frac{2}{3}\right)^{i+1} O(m\alpha(n)\log (\rho 2^{i+2} )) \\
&\leq 3O(m)+3 O(m\alpha(n)\log(\rho 2^\Gamma ))\\
&=O(m\alpha(n)\log(m/n )).
\end{align*}
And if $m\log (n)< m\alpha(n)\log(m/n)$, we have that w.h.p.\ the first phase takes total time at most
\begin{align*}
\sum_{i=0}^\Gamma  O(|X_{i}|)+O(|X_{i+1}|\log (n) ) &= \sum_{i=0}^\Gamma  \left(\frac{2}{3}\right)^{i}O(m)+\left(\frac{2}{3}\right)^{i+1} O(m\log (n) ) \\
&\leq 3O(m)+3 O(m\log (n))\\
&=O(m\log n).
\end{align*}

In the second phase, we sample each edge $e$ from the binomial distribution with parameters $n_e$ and $p_e$. We will show this can be done with a process that takes $T=O(m)$ time with high probability. Hereto, we use an algorithm from \cite{D80} for binomial sampling, for which the pseudocode is given in Algorithm~\ref{alg:binomsampling}. 

\vspace{1em}
\begin{algorithm}[hbt!]
\SetAlgoLined \caption{\textsc{Binom}$(n,p)$} \label{alg:binomsampling}
\KwIn{Two parameters $n,p$.}
\KwOut{A random sample from the binomial distribution with parameters $n$ and $p$.}
Set $k\leftarrow -1$, $S\leftarrow 0$.\\
\While{$S<$n}{
$k\leftarrow k+1$.\\
Generate $u \sim \mathcal U(0,1)$.\\
$S \leftarrow S+\lfloor \log (u) /\log(1-p)\rfloor +1.$
}
\Return{$k$}
\end{algorithm}
\vspace{1em}

It is easy to see that this algorithm takes $O(1+k)$ time, where $k$ is the output. So if the sample from the binomial distribution is $k$, this takes time $O(1+k)$. This means that the total time $T$ equals $m$ plus the total sum of all samples. Note that this is slightly different from what we did in Section~\ref{subsc:size} to bound the number of edges: there we needed to bound the number of positive samples. 

For each edge $e\in F_i$ we need to draw from the binomial distribution with parameters $n_e$ and $p_e$. We denote $T_e$ for the time we need to sample $e$. By the above, we have $\E[T_e] = 1+n_ep_e$. 
So, the expected number of successes is at most 
\begin{align*}
\E[T]=\sum_i\sum_{e\in F_i}\E[T_e]=\sum_i\sum_{e\in F_i}(1+n_ep_e)= \sum_i |F_i|+O(cn\log (n)\log\left(m/(n\log (n)/\epsilon^2)\right)/\epsilon^2),
\end{align*} 
as shown in Section~\ref{subsc:size}. Let $B>0$ such that $\sum_i\sum_{e\in F_i} n_ep_e\leq Bcn\ln (n)\log\left(m/(n\log (n)/\epsilon^2)\right)/\epsilon^2$. We can use a Chernoff bound (see Theorem~\ref{thm:Chernoffupperbound}) on the sum of these $\sum_i\sum_{e\in F_i}n_e$ random variables to obtain:
\begin{align*}
&\P\left[ T-\sum_i |F_i|>2Bcn \ln (n)\log\left(m/(n\log (n)/\epsilon^2)\right)/\epsilon^2\right] \\
&\leq \exp\left(-0.38 Bcn \ln (n)\log\left(m/(n\log (n)/\epsilon^2)\right)/\epsilon^2\right) \\
&= n^{-0.38 Bcn\log\left(m/(n\log (n)/\epsilon^2)\right)/\epsilon^2}. 
\end{align*}
So we can say that with high probability we need 
\begin{align*}
T = \sum_i |F_i|+ \left(T-\sum_i |F_i|\right) = O(m)+O(2Bcn\ln (n)\log\left(m/(n\log (n)/\epsilon^2)\right)/\epsilon^2)=O(m)
\end{align*} time for the sampling. 

Concluding, the algorithm takes $$O(m\cdot\min(\alpha(n)\log (m/n),\log (n) )+O(m)= O(m\cdot\min(\alpha(n)\log (m/n),\log (n) )$$ time in total for polynomially-weighted graphs. 

\section{Adaptation to Unbounded Weights}
\label{sc:arbweights}
In this section, we sketch how we can adapt the algorithm of the previous section to sparse graphs with unbounded weights. The key to this is Lemma~\ref{lm:MSFA}, which shows that for unbounded weights we might not be able to compute the MSF indices exactly, but we can find an estimate for edges $e$ with  $w(e) > d(e)/n$. Recall the definition of $d(e)$: compute a single maximum spanning forest $F$ for $G$ and define $d(e)$ to be the minimum weight among the edges on the path from $u$ to $v$ in $F$, where $e=(u,v)$.

The only adaptation for unbounded weights is that the first time we compute maximum spanning forests in Algorithm~\ref{alg:mainalg}, we set aside any edges $e\in E$ with $w(e) \leq d(e)/n$. We show that we can sample efficiently from these vertices, since they are well-connected by $F_0$, the initial MSF that remains in our sparsifier. We will do this by sampling them with $\lambda_e=\rho\cdot d(e)$. Note that we only have to set aside vertices the first time we compute a MSF packing, after this the estimates $d(e)$ in a new graph can only decrease, so if a vertex satisfies $w(e) \leq d(e)/n$ in a certain subgraph, it also satisfied this in the initial graph. 

For the remaining vertices, we apply the algorithm as presented in the previous section. The only difference is that we use Lemma~\ref{lm:MSFA} to compute an estimate of the MSF indices. This means that if an edge $e\in E$ obtains the estimate index $\tilde{f}_e$ w.r.t. some graph $E'$, we have that $e$ is at least $f_ew_e(1-1/n)$-heavy in $E'$. For simplicity, we use $1-1/n\geq 1/2$. We see that this impacts the analysis in two places where the heaviness is used: Lemma~\ref{lm:onestepcutpreservationweightclasses} and Lemma~\ref{lm:RjkheavyinEjk}. 

We examining Lemma~\ref{lm:onestepcutpreservationweightclasses}, we see that we apply Lemma~\ref{lm:chernoffforheavyedges} with $\delta^2p\pi \geq \frac{\zeta\ln (n)}{0.38}$, for certain $\delta, p, \pi$, and $\zeta$. We want to apply this lemma but have $\tilde{\pi}=\pi/2$, hence we set $\tilde{\delta} = \sqrt{2}\delta$. If we want to end up with the original result of Lemma~\ref{lm:onestepcutpreservationweightclasses}, we set the $\tilde{\epsilon} = \epsilon/\sqrt{2}$. This constant factor change gets absorbed in the asymptotic notation for size and running time of the algorithm.

The second lemma we investigate is Lemma~\ref{lm:RjkheavyinEjk}, which is the $\Pi$-connectivity in the sampling. Here, there is an easy solution: we boost all edges in $E_{j,k}$ by a factor two, which ensures the $\Pi$-connectivity as desired. Consequently, all edges in $E_i$ are boosted with a factor two, which propagates to a factor two in $e_i(C)$ as denoted in Lemma~\ref{lm:gammaoverlap}, resulting to a $\gamma$-overlap with $\gamma =\frac{128}{3}$, rather than $\frac{64}{3}$. 

Summing this up, we can say that our original analysis holds when we call the algorithm with $\tilde{\epsilon}=\epsilon/\sqrt{2}$ and $\tilde{\rho}= \frac{(7+c)2704 \ln (n)}{0.38 \epsilon^2}$, where the change in $\rho$ is a direct consequence of the change in $\gamma$. 

The last thing that remains, is to show that, when we sample, $\Pi$-connectivity is also satisfied for the edges $e\in E$ with $w(e) \leq d(e)/n$. This is an extension to Corollary~\ref{cor:piconnectivity}. 

\begin{lemma}
Suppose $e \in R_i$ and $w(e) \leq d(e)/n$, then $e$ is $\pi = \rho\cdot4^\Gamma2^\Lambda$-heavy in $G_i=(V,E_i)$, with $E_i=\bigcup_{j=1}^{\min(\lfloor i/2\rfloor,\Gamma)}E_{j,i-2j}$. 
\end{lemma} 
\begin{proof}
We know that $e$ is $d(e)$-heavy in $F_0$, so we look for the occurrence of $F_0$ in $E_i$:
\begin{align}\label{eq:F_0inE_i}
E_i &= \bigcup_{j=1}^{\min(\lfloor i/2\rfloor,\Gamma)}\bigcup_{j'=j-1}^\Gamma \bigcup_{k'=i-2j}^\infty \rho\cdot 4^{\Gamma-j'+1}2^{\Lambda-k'+j'}\{e'\in F_{j'} : 2^{k'} \leq \rho\cdot w(e') \leq 2^{k'+1}-1\}\nonumber \\
&\supseteq  \rho\cdot 4^{\Gamma+1} \bigcup_{k'=i-2}^\infty 2^{\Lambda-k'} \{e'\in F_0 : 2^{k'} \leq \rho\cdot w(e') \leq 2^{k'+1}-1\}.
\end{align}
We look more closely at the connectedness of $e$ in this particular set. We note that $w(e')\geq d(e)$ for any edge on a path in $F_0$ from $u$ to $v$ for $e=(u,v)$, by definition of $d(e)$. So we only need to consider $e'\in F_0$ with $\rho\cdot w(e')\geq \rho \cdot d(e)=\lambda_e \geq 2^i$, as $e\in R_i$. This means that $e$ is $d(e)$-heavy in 
\begin{align*}
\bigcup_{k'=i}^\infty \{e'\in F_0 : 2^{k'} \leq \rho\cdot w(e') \leq 2^{k'+1}-1\}\subseteq  \bigcup_{k'=i-2}^\infty \{e'\in F_0 : 2^{k'} \leq \rho\cdot w(e') \leq 2^{k'+1}-1\}.
\end{align*}
We can rescale this to exploit the weights fully: $e$ is $2^\Lambda$-heavy in  $ \bigcup_{k'=i-2}^\infty 2^{\Lambda-k'} \{e'\in F_0 : 2^{k'} \leq \rho\cdot w(e') \leq 2^{k'+1}-1\}$. Combining this with Equation~\ref{eq:F_0inE_i} gives us that $e$ is $\rho \cdot4^{\Gamma+1}2^\Lambda$-heavy in $E_i$, which is a factor four more than we needed to show. 
\end{proof}

\subsection{Size and time complexity}
For the size of the resulting graph $G_\epsilon$, the upper bound of the previous section still holds for the edges that are sampled according to their MSF index. It remains to show that the contribution of any edges with $w(e) \leq d(e)/n$ is small. For these edges we have $p_e = \frac{384}{169}\frac{1}{d(e)}$. We use $\P[ \text{Binom}(n_e,p_e) > 0] \leq p_e n_e$, to see
\begin{align*}
\P[ e\in G_\epsilon : \lambda_e = \rho \cdot d(e)] \leq \frac{384}{169}\frac{w(e)}{d(e)} \leq  \frac{384}{169}\frac{w(e)}{n\cdot w(e)}.
\end{align*}
As there can be at most $n^2$ edges with $w(e)\leq d(e)/n$, we obtain that the expected number of edges in $G_\epsilon$ originating from such edges is at most $O(n)$. By the same arguments as given in Section~\ref{subsc:size}, this holds not only in expectation, but also with high probability. 

Concerning the time complexity, we use Theorem~\ref{thm:MSFdense} or Lemma~\ref{lm:MSFA} instead of Theorem~\ref{thm:MSFP}. These run in time $O(m\log(n))$ and $O(m\alpha(n)\log (M))$ respectively. Since the size of the sparsifier does not increase significantly, the time needed for sampling does not increase significantly either. Hence we obtain a total time of $O(m\cdot\min(\alpha(n)\log(m/n),\log(n)))$. This makes the algorithm the fastest cut sparsification algorithm known for graphs with unbounded weights.

\section{Conclusion}
In this paper, we presented a faster $(1\pm\epsilon)$-cut sparsification algorithm for weighted graphs. We have shown how to compute sparsifiers of size $O(n\log (n)/\epsilon^2)$ in $O(m\cdot\min(\alpha(n)\log (m/n), \log (n)))$ time, for integer weighted graphs. Both algorithms apply a sampling technique where the MSF index is used as a connectivity estimator. 

We have shown that we can compute an $M$-partial MSF packing in $O(m\alpha(m)\log (M))$ time for polynomially-weighted graphs. For graphs with unbounded integer weights, we have shown that we can compute a complete MSF packing in $O(m\log (n))$ time, and a sufficient estimation of an $M$-partial MSF packing can be computed in time $O(m\alpha(m)\log (M))$. An open question is whether a more efficient computation is possible. This would improve on our sparsification algorithm, but might also be advantageous in other applications. The NI index has shown to be useful in various applications. We believe to have shown that the MSF index is a natural analogue.

To develop an algorithm to compute an MSF packing, one might be inclined to build upon one of the algorithms that compute a minimum spanning tree faster than Kruskal's algorithm, such as the celebrated linear-time algorithm of Karger, Klein, and Tarjan \cite{KKT95}. However, this algorithm and many other fast minimum spanning tree algorithms make use of edge contractions. It is far from obvious how to generalize this to a packing: in that case, we need to work simultaneously on multiple trees, hence we cannot simply contract the input graph in favor of any single one. To make this work, a more meticulous use of data structures seems necessary. 

Computation of the MSF indices in linear time would be an ultimate goal. However, for our application a slightly looser bound suffices. If we can reduce the running time to compute the MSF indices to $O(m+n\log (n))$, then we obtain a time bound of $O(m)$ for cut sparsification. Moreover, we do not need the exact MSF index, an estimate suffices. This can either be a constant-factor approximation of the MSF index for each edge, or an estimate in the weights used in the forests, as done for graphs with unbounded weights in Section~\ref{sc:arbweights}.

\bibliographystyle{alpha}
\bibliography{references}

\appendix
\section{Tail bounds}
To analyze the sampling methods used in Section~\ref{sc:ouralg}, we make use of the well-known Chernoff bound to get a grasp on the tail of various distributions \cite{chernoff}. 
\begin{theorem}
\label{thm:Chernoff}
Let $Y_1, \dots, Y_n$ be $n$ independent random variables such that each $Y_i$ takes values in $[0,1]$. Let $\mu=\sum_{i=1}^n \E[Y_i]$ and $\xi=2\ln (2)>0.38$. Then for all $\epsilon>0$
\[ \P\left[\left|\sum_{i=1}^nY_i-\mu\right|>\epsilon \mu\right] \leq 2\exp\left(-\xi\min(\epsilon,\epsilon^2)\mu\right).\]
\end{theorem}
At times, the expected value $\mu$ itself is not known. Fortunately an upper bound on the expected value also suffices. 
\begin{theorem}
\label{thm:Chernoffupperbound}
Let $Y_1, \dots, Y_n$ be $n$ independent random variables such that $Y_i$ takes values in $[0,1]$. Let $\mu=\sum_{i=1}^n \E[Y_i]$ and $\xi=2\ln (2)>0.38$. Suppose $\mu'\geq \mu$. Then for all $\delta\geq2$
\[ \P\left[\sum_{i=1}^nY_i>\delta \mu'\right] \leq 2\exp\left(-\xi(\delta-1) \mu'\right).\]
\end{theorem}
\begin{proof}
Let $\epsilon := (\delta -1) \frac{\mu'}{\mu}$. We have $\epsilon\geq 1$, so $\min(\epsilon,\epsilon^2)=\epsilon$. The statement now follows directly from Theorem~\ref{thm:Chernoff}.
\end{proof}

\section{Reduction from Real to Integer Weights}\label{app:reduction}
In this section, we show how to reduce the computation of a cut sparsifier of a graph with non-negative real weights to integer weights, formalizing the procedure sketched by Benczúr and Karger~\cite{BK15}. Let $G=(V,E,w)$ be a weighted graph, where $w\colon E\to \R$. Denote $W_{\rm{max}}:= \max\limits_{e\in E} w(e)$ and $W_{\rm{min}}:= \min\left\{1,\min\limits_{e\in E} w(e)\right\}$. Then the reduction consists of the following steps:
\begin{enumerate}
	\item Compute $W_{\rm{min}}$ and $r:= -\lfloor \log(\tfrac{\epsilon}{2}W_{\rm{min}})\rfloor$. \label{step:defr}
	\item Create $w'\colon E \to \R$ by rounding the weights $w(e)$ to the closest multiple of $2^{-r}$, and define $G':=(V,E,w')$.\label{step:round}
	\item Create $\hat w\colon E \to \R$ by $\hat w(e) :=2^r w'(e)$. \label{step:scale}
	\item Compute a $(1\pm \epsilon/3)$-cut sparsifier $\hat H=(V,E_H,\hat w_H)$ of $\hat G=(V,E,\hat w)$.\label{step:reweighted}
	\item Output $H=(V,E_H, w_H)$ where $w_H(e) := 2^{-r}\hat w_H(e)$. \label{step:scaleback}
\end{enumerate}

First, we show that the graph $H$ is indeed a $(1+\epsilon)$-cut sparsifier of $G$. Hereto, we note that for any cut $C$ we have
\begin{align*}
	w_H(C) &= 2^{-r} w_{ \hat H}(C) \leq 2^{-r} (1+\epsilon/3)  w_{\hat G}(C) = (1+\epsilon/3) w_{G'}(C)
	\leq (1+\epsilon) w_G(C),
\end{align*}
where the last inequality holds as each weight $w'(e)$ has at most an additive error of $2^{-r}\leq \tfrac{\epsilon}{2}W_{\rm{min}}\leq \tfrac{\epsilon}{2}$ with respect to $w(e)$, hence at most an multiplicative error of $\tfrac{\epsilon}{2}$. Analogously we obtain $w_H(C)\geq (1-\epsilon) w_G(C)$. 

By construction, $\hat G$ has integer weights, which are bounded by $O(\tfrac{W_{\rm{max}}}{\epsilon W_{\rm{min}}})$. Steps~\ref{step:defr}, \ref{step:round}, \ref{step:scale}, and~\ref{step:scaleback} can be implemented in $O(m)$ time. So indeed we have reduced the problem to finding a cut sparsifier of a graph with integer weights. 
Moreover, note that if $G$ has polynomially bounded real weights, in the sense that $W_{\rm{max}}= O(\poly(n))$ and $W_{\rm{min}} = \Omega(1/\poly(n))$, then the graph $\hat G$ has polynomially bounded integer weights. We can state this independent of $\epsilon$, since for $\epsilon\leq 1/m$ we can always output the entire input graph as a cut sparsifier of optimal size $O(n/\epsilon^2)$~\cite{ACK+16}. 
\end{document}